\documentclass{LMCS}
\def\doi{7 (3:19) 2011}
\lmcsheading%
{\doi}
{1--37}
{}
{}
{Dec.~\phantom07, 2010}
{Sep.~28, 2011}
{}   

\usepackage{amsmath}
\usepackage{amssymb}
\usepackage{amsthm}
\usepackage{graphicx}
\usepackage{tikz}
\usepackage{url}
\usepackage{enumerate,hyperref}

\renewcommand{\paragraph}{\textit}

\def\tab{\hspace*{5mm}}

\newcommand{\sat}{SAT}
\newcommand{\smt}{SMT}
\newcommand{\dpll}{DPLL}

\newcommand{\fn}{\sf}

\renewcommand{\implies}{\ \longrightarrow\ }
\renewcommand{\iff}{\mathrm{\ iff\ }}
\newcommand{\falsifies}{\vDash\!\!\lnot\,}
\newcommand{\nfalsifies}{\nvDash\!\!\lnot\,}
\newcommand{\impliesLiteral}{\vDash}
\newcommand{\impliesClause}{\vDash}
\newcommand{\impliesFormula}{\vDash}
\newcommand{\impliesValuation}{\vDash}
\newcommand{\entailsLiteral}{\vDash}
\newcommand{\entailsClause}{\vDash}

\newcommand{\opp}[1]{\overline{#1}}

\newcommand{\restrict}[2]{#2\!\!|_{#1}}
\newcommand{\lexprod}[2]{#1 \,\langle\mathrm{*lex*}\rangle\, #2}
\newcommand{\lexprodp}[2]{#1 \,\langle\mathrm{*lex^p*}\rangle\, #2}
\newcommand{\lex}[1]{#1^{\mathrm{lex}}}
\newcommand{\mult}[1]{#1^{\mathrm{mult}}}
\newcommand{\mults}[1]{#1^{\mathrm{mult}_1}}
\newcommand{\multisetof}[1]{\langle #1\rangle}
\newcommand{\transclosure}[1]{#1^+}
\newcommand{\precedes}[1]{\prec_{#1}}
\newcommand{\succlit}{\succ_{\mathrm{lit}}}
\newcommand{\succTr}{\succ_{\mathrm{tr}}}
\newcommand{\succTrrestrict}[1][Vbl]{\restrict{#1}{\succTr}}
\newcommand{\succBool}{\succ_{\mathrm{bool}}}
\newcommand{\succF}[1]{\succ_{\mathrm{Form}}^{#1}}
\newcommand{\succFIncl}{\succ_{\mathrm{Form}\subset}}
\newcommand{\succFInclrestrict}[1]{\restrict{#1}\succFIncl}
\newcommand{\succC}[1]{\succ_{\mathrm{Cla}}^{#1}}

\newcommand{\elements}{}
\newcommand{\valform}[1]{\langle #1 \rangle}
\newcommand{\decision}[1]{#1^d}
\newcommand{\nondecision}[1]{#1^i}

\usepackage{stmaryrd}

\newcommand{\rightarrowd}{\rightarrow_d}
\newcommand{\rightarrowb}{\rightarrow_b}
\newcommand{\rightarrowl}{\rightarrow_l}
\newcommand{\rightarrowc}{\rightarrow_c}
\newcommand{\rightarrowr}{\rightarrow_r}
\newcommand{\rightarrowf}{\rightarrow_f}

\newcounter{tmpcounter}

\newcounter{definitioncounter}

\newtheorem{proposition}{Proposition}
\newcounter{theoremcounter}

\newtheoremstyle{example}{\topsep}{\topsep}%
     {}
     {}
     {\bfseries}
     {.}
     {1em}
     {\thmname{#1}\thmnumber{ #2}\thmnote{ #3}}
\theoremstyle{example}

\newcommand{\rulename}[1]{{\sf\small #1}}

\newcounter{inv:counter}
\newcommand{\defineInvariant}[1]{\refstepcounter{inv:counter}\label{#1}}
\newcommand{\inv}[2][]{\ensuremath{Inv_{#2#1}}}

\newcommand{\packitems}{%
\setlength{\itemsep}{0pt}
\setlength{\parsep}{0pt}
\setlength{\parskip}{0pt}}

\begin{document}
\title{Formalization of Abstract State Transition Systems for SAT}
\author[F.~ Mari\'c]{Filip Mari\'c}
\address{Faculty of Mathematics, University of Belgrade, Studentski Trg 16, 11000 Belgrade}
\email{filip@matf.bg.ac.rs, janicic@matf.bg.ac.rs}
\author[P.~Jani\v{c}i\'c]{Predrag Jani\v{c}i\'c}

\begin{abstract}
We present a formalization of modern \sat{} solvers and their properties in a
form of \emph{abstract state transition systems}. \sat{} solving procedures are
described as transition relations over states that represent the values of
the solver's global variables. Several different \sat{} solvers are formalized,
including both the classical DPLL procedure and its state-of-the-art successors.
The formalization is made within the Isabelle/HOL system and the total
correctness (soundness, termination, completeness) is shown for each presented
system (with respect to a simple notion of satisfiability that can be manually
checked).  The systems are defined in a general way and cover procedures used in
a wide range of modern \sat{} solvers. Our formalization builds up on the
previous work on state transition systems for \sat{}, but it gives
machine-verifiable proofs, somewhat more general specifications, and weaker
assumptions that ensure the key correctness properties.  The presented proofs of
formal correctness of the transition systems can be used as a key building block
in proving correctness of \sat{} solvers by using other verification approaches.
\end{abstract}
\maketitle


\section{Introduction}
\label{sec:introduction}

The problem of checking propositional satisfiability (\sat) is one of the
central problems in computer science. It is the problem of deciding if there is
a valuation of variables under which a given propositional formula (in
conjunctive normal form) is true. \sat{} was the first problem that was proved
to be {\sc NP}-complete \cite{Cook71} and it still holds a central position in
the field of computational complexity. \sat{} solvers, procedures that solve
the \sat{} problem, are successfully used in many practical applications such as
electronic design automation, software and hardware verification, artificial
intelligence, and operations research.

Most state-of-the-art complete \sat{} solvers are essentially based on a branch
and backtrack procedure called Davis-Putnam-Logemann-Loveland or the \dpll{}
procedure \cite{dp60,dll62}. Modern \sat{} solvers usually also
employ~(i)~several conceptual, high-level algorithmic additions to the original
\dpll{} procedure, (ii) smart heuristic components, and (iii) better low-level
implementation techniques. Thanks to these, spectacular improvements in the
performance of \sat{} solvers have been achieved and nowadays \sat{} solvers can
decide satisfiability of CNF formulae with tens of thousands of variables and
millions of clauses.

The tremendous advance in the \sat{} solving technology has not been accompanied
with corresponding theoretical results about the solver
correctness. Descriptions of new procedures and techniques are usually given in
terms of implementations, while correctness arguments are either not given or
are given only in outlines. This gap between practical and theoretical progress
needs to be reduced and first steps in that direction have been made only
recently, leading to the ultimate goal of having modern \sat{} solvers that are
formally proved correct. That goal is vital since \sat{} solvers are used in
applications that are very sensitive (e.g., software and hardware verification)
and their misbehaviour could be both financially expensive and dangerous from
the aspect of security.  Ensuring trusted \sat{} solving can be achieved by two
approaches.

One approach for achieving a higher level of confidence in SAT solvers' results,
successfully used in recent years, is proof-checking
\cite{zhangmalik-validating,clausal-proofs,van-gelder-verifyingunsat,weber,sat-coq}. In
this approach, solvers are modified so that they output not only {\em sat} or
{\em unsat} answers, but also justification for their claims (models for
satisfiable instances and proof objects for unsatisfiable instances) that are
then checked by independent proof-checkers. Proof-checking is relatively easy to
implement, but it has some drawbacks. First, justification for every solved
\sat{} instance has to be verified separately. Also, generating unsatisfiability
proofs introduces some overhead to the solver's running time, proofs are
typically large and may consume gigabytes of storage space, and proof-checking
itself can be time consuming \cite{van-gelder-verifyingunsat}. Since
proof-checkers have to be trusted, they must be very simple programs so they can
be ``verified'' by code inspection.\footnote{Alternatively, proof-checkers could
  be formally verified by a proof assistant, and then their correctness would
  rely on the correctness of the proof assistant.}  On the other hand, in order
to be efficient, they must use specialized functionality of the underlying
operating system which reduces the level of their reliability (e.g., the proof
checker used in the \sat{} competitions uses Linux's mmap functionality
\cite{van-gelder-verifyingunsat}).

The other approach for having trusted solvers' results is to verify the \sat{}
solver itself, instead of checking each of its claims. This approach is very
demanding, since it requires formal analysis of the complete solver's behaviour.
In addition, whenever the implementation of the solver changes, the correctness
proofs must be adapted to reflect the changes. Still, in practice, the core
solving procedure is usually stable and stays fixed, while only heuristic
components frequently change. The most challenging task is usually proving the
correctness of the core solving procedures, while heuristic components only need
to satisfy relatively simple properties that are easily checked.  This approach
gives also the following benefits:

\begin{iteMize}{$\bullet$}
\packitems
\item Although the overheads of generating unsatisfiability proofs during
  solving are not unmanageable, in many applications they can be avoided if the
  solver itself is trusted.\footnote{In some applications, proofs of
    unsatisfiability are still necessary as they are used, for example, for
    extracting unsatisfiable cores and interpolants.}
\item Verification of modern \sat{} solvers could help in better theoretical
  understanding of how and why they work. A rigorous analysis and verification
  of modern \sat{} solvers may reveal some possible improvements in underlying
  algorithms and techniques which can influence and improve other solvers as
  well.
\item Verified \sat{} solvers can serve as trusted kernel checkers for verifying
  results of other untrusted verifiers such as BDDs, model checkers, and SMT
  solvers. Also, verification of some \sat{} solver modules (e.g., Boolean
  constraint propagation) can serve as a basis for creating both verified and
  efficient proof-checkers for \sat{}.
\end{iteMize}

\noindent In order to prove correctness of a \sat{} solver, it has to be formalized in
some meta-theory so its properties can be analyzed in a rigorous mathematical
manner. In order to achieve the desired highest level of trust, formalization
in a classical ``pen-and-paper'' fashion is not satisfactory and, instead, a
mechanized and machine-checkable formalization is preferred. The formal
specification of a \sat{} solver can be made in several ways (illustrated in
Figure \ref{fig:SATproject}, each with an appropriate verification paradigm and
each having its own advantages and disadvantages, described in the following
text).

\begin{figure}[t!]
\begin{center}
\input{SATproject.tkz}
\end{center}
\vspace*{-4mm}
\caption{Different approaches for \sat{} solver verification}
\label{fig:SATproject}
\end{figure}

\begin{desCription}
\item\noindent{\hskip-12 pt\bf Verification of abstract state transition
system:}\
  State transition
  systems are an abstract and purely mathematical way of specifying program
  behaviour. Using this approach, the \sat{} solver's behaviour is modelled by
  transitions between states that represent the values of the solver's global
  variables. Transitions can be made only by following precisely defined
  transition rules. Proving correctness of state transition systems can be
  performed by the standard mathematical apparatus. There are state transition
  systems describing the top-level architecture of the modern \dpll{}-based
  \sat{} solvers (and related \smt{} solvers) \cite{Krstic-Frocos07,
    NieOT-Jacm06} and their correctness has been informally shown.

  The main advantage of the abstract state transition systems is that they are
  mathematical objects, so it is relatively easy to make their formalization
  within higher-order logic and to formally reason about them. Also, their
  verification can be a key building block for other verification approaches.
  Disadvantages are that the transition systems do not specify many details
  present in modern solver implementations and that they are not directly
  executable.

\item\noindent{\hskip-12 pt\bf Verified implementation within a proof
  assistant:}\
  A program's behavi\-our
  can be specified within the higher-order logic of a proof assistant (regarded
  as a purely functional programming language). This approach is often called
  {\em shallow embedding into HOL}. Specifications may vary from very abstract
  ones to detailed ones covering most details present in the real \sat{}
  solver's code. The level of details can incrementally be increased (e.g., by
  using a datatype refinement). Having the specification inside the logic, its
  correctness can be proved again by using the standard mathematical apparatus
  (mainly induction and equational reasoning). Based on the specification,
  executable functional programs can be generated by means of code extraction
  --- the term language of the logic within the proof assistant is identified
  with the term language of the target language and the verified program
  correctness is transferred to the exported program, up to simple
  transformation rules.

  Advantages of using the shallow embedding are that, once the solver is defined
  within the proof assistant, it is possible to verify it directly inside the
  logic and a formal model of the operational or denotational semantics of the
  language is not required. Also, extracted executable code can be trusted with
  a very high level of confidence. On the other hand, the approach requires
  building a fresh implementation of a \sat{} solver within the logic. Also,
  since higher-order logic is a pure functional language, it is unadapted
  to modelling imperative data-structures and their destructive updates. Special
  techniques must be used to have mutable data-structures and, consequently, an
  efficient generated code \cite{imperativeHOL}.

\item\noindent{\hskip-12 pt\bf Verification of the real implementations:}\
  The most demanding approach for
  verifying a \sat{} solver is to directly verify the full real-world solver
  code. Since \sat{} solvers are usually implemented in imperative programming
  languages, verifying the correctness of implementation can be made by using
  the framework of Hoare logic \cite{hoare} --- a formal system for reasoning
  about programs written in imperative programming languages.  The program
  behaviour can then be described in terms of preconditions and postconditions
  for pieces of code. Proving the program correctness is made by formulating and
  proving verification conditions. For instance, Isabelle/HOL provides a formal
  verification environment for sequential imperative programs (\cite{schirmer}).

  The main benefit of using the Hoare style verification is that it enables
  reasoning about the imperative code, which is the way that most real-world
  \sat{} solvers are implemented. However, since real code is overwhelmingly
  complex, simpler approximations are often made and given in pseudo-programming
  languages. This can significantly simplify the implementation, but leaves a
  gap between the correctness proof and the real implementation.
\end{desCription}

\noindent In this paper we focus on the first verification approach as it is often
suitable to separate the verification of the abstract algorithms and that of
their specific implementations.\footnote{A recent example is the L4 verified OS
kernel, where a shallowly embedded Haskell specification of the kernel is
verified, and then the C code is shown to implement the Haskell specification,
yielding a natural separation of concepts and issues \cite{Klein10}.}  In
addition, state transition systems, as the most abstract specifications, cover
the widest range of existing \sat{} solver implementations. Moreover, the
reasoning used in verifying abstract state transition systems for \sat{} can
serve as a key building block in verification of more detailed descriptions of
\sat{} solvers using the other two approaches described above (as illustrated by
Figure \ref{fig:SATproject}). Indeed, within our \sat{} verification project
\cite{TPHOLS}, we have already applied these two approaches
\cite{JAR,TCS,Informatica}, and in both cases the correctness arguments were
mainly reduced to correctness of the corresponding abstract state transition
systems.  These transition systems and their correctness proofs are presented in
this paper for the first time, after they evolved to some extent through
application within the other two verification approaches.

The methodology that we use in this paper for the formalization of \sat{}
solvers via transition systems is \emph{incremental refinement}: the
formalization begins with a most basic specification, which is then refined by
introducing more advanced techniques, while preserving the correctness. This
incremental approach proves to be a very natural approach in formalizing complex
software systems. It simplifies understanding of the system and reduces the
overall verification effort. Each of the following sections describes a separate
abstract state transition system. Although, formally viewed, all these systems
are independent, each new system extends the previous one and there are tight
connections between them. Therefore, we do not expose each new system from
scratch, but only give additions to the previous one. We end up with a system
that rather precisely describes modern \sat{} solvers, including advanced
techniques such as backjumping, learning, conflict analysis, forgetting and
restarting. The systems presented are related to existing solvers, their
abstract descriptions and informal correctness proofs.

The paper is accompanied by a full formalization developed within the
Isabelle/HOL proof assistant.\footnote{The whole presented formalization is
  available from AFP \cite{afp-sat} and, the latest version, from
  \url{http://argo.matf.bg.ac.rs}.}  The full version of the paper\footnote{The
  full version of the paper is available from \url{http://argo.matf.bg.ac.rs}.}
contains an appendix with informal proofs of all lemmas used.  All definitions,
lemmas, theorems and proofs of top-level statements given in the paper
correspond to their Isabelle counterparts, and here are given in a form
accessible not only to Isabelle users, but to a wider audience.

The main challenge in each large formalization task is to define basic relevant
notions in appropriate terms, build a relevant theory and a suitable hierarchy
of lemmas that facilitates constructing top-level proofs. Although in this paper
we do not discuss all decisions made in the above directions, the final
presented material is supposed to give the main motivating ideas and,
implicitly, to illustrate a proof management technology that was used. The main
purpose of the paper is to give a clear picture of central ideas relevant for
verification of \sat{} transition systems, hopefully interesting both to \sat{}
developers and to those involved in formalization of mathematics.

The main contributions of this paper are the following.
\begin{iteMize}{$\bullet$}
\item \sat{} solving process is introduced by a hierarchical series of abstract
  transition systems, ending up with the state-of-the-art system.
\item Formalization and mechanical verification of properties of the abstract
  transition systems for \sat{} are performed (within this, invariants and
  well-founded relations relevant for termination are clearly given; conditions
  for soundness, completeness, and termination are clearly separated). Taking
  advantage of this formalization, different real-world \sat{} solvers can be
  verified, using different verification approaches.
\item First proofs (either informal or formal) of some properties of modern
  \sat{} solvers (e.g., termination condition for frequent restarting) are
  given, providing deeper understanding of the solving process.
\end{iteMize}

\noindent The rest of the paper is organized as follows: In Section \ref{sec:Background}
some background on \sat{} solving, abstract state transition systems, and
especially abstract state transition systems for \sat{} is given. In Section
\ref{sec:backgroundtheory} basic definitions and examples of propositional logic
and CNF formulae are given. In Section \ref{sec:DPLLSearch}, a system
corresponding to basic \dpll{} search is formalized.  In Section
\ref{sec:backjumping}, that system is modified and backtracking is replaced by
more advanced backjumping. In Section \ref{sec:learningForgetting}, the system
is extended by clause learning and forgetting. In Section
\ref{sec:conflictAnalysis} and Section \ref{sec:restart} a system with conflict
analysis and a system with restarting and forgetting are formalized. In Section
\ref{sec:discuss} we discuss related work and our contributions.  In Section
\ref{sec:conclusions}, final conclusions are drawn.

\section{Background}
\label{sec:Background}

In this section we give a brief, informal overview of the \sat{} solving
process, abstract state transition systems and abstract state transition systems
for \sat{}. The paper does not intend to be a tutorial on modern \dpll-based
\sat{} solving techniques --- the rest of the paper contains only some brief
explanations and assumes the relevant background knowledge (more details and
tutorials on modern \sat{} solving technology can be found in other sources
e.g., \cite{sathandbook,JAR}).

\subsection{\sat{} Solving}
\label{subsec:backround_solving}

\sat{} solvers are decision procedures for the satisfiability problem for
propositional formulae in conjunctive normal form (CNF). State-of-the-art \sat{}
solvers are mainly based on a branch-and-backtrack procedure called \dpll{}
(Davis-Putnam-Logemann-Loveland) \cite{dp60,dll62} and its modern
successors. The original \dpll{} procedure (shown in Figure \ref{fig:dpll})
combines backtrack search with some basic, but efficient inference rules.

\begin{figure}[ht]
\begin{flushleft}
\begin{footnotesize}
\texttt{%
function {\fn dpll} ($F$ : Formula) : (SAT, UNSAT)\\
begin\\
\tab if $F$ is empty then return SAT\\
\tab else if there is an empty clause in $F$ then return UNSAT\\
\tab else if there is a pure literal $l$ in $F$ then return {\fn dpll}($F[l\rightarrow \top]$)\\
\tab else if there is a unit clause $[l]$ in $F$ then return {\fn dpll}($F[l\rightarrow \top]$)\\
\tab else begin\\
\tab \tab select a literal $l$ occurring in $F$\\
\tab \tab if {\fn dpll}($F[l\rightarrow \top]$) = SAT then return SAT\\
\tab \tab else return {\fn dpll}($F[l\rightarrow \bot]$)\\
\tab end\\
end}
\end{footnotesize}
\end{flushleft}
\caption{The original DPLL procedure} \label{fig:dpll}
\end{figure}

\noindent The search component selects a branching literal $l$ occurring in the formula
$F$, and tries to satisfy the formula obtained by replacing $l$ with $\top$ and
simplifying afterwards. If the simplified formula is satisfiable, so is the
original formula $F$. Otherwise, the formula obtained from $F$ by replacing $l$
with $\bot$ and by simplifying afterwards is checked for satisfiability and it
is satisfiable if and only if the original formula $F$ is satisfiable.  This
process stops if the formula contains no clauses or if it contains an empty
clause. A very important aspect of the search process is the strategy for
selecting literals for branching --- while not important for the correctness of
the procedure, this strategy can have a crucial impact on efficiency.

The simple search procedure is enhanced with several simple inference
mechanisms. The \emph{unit clause} rule is based on the fact that if there is a
clause with a single literal present in $F$, its literal must be true in order
to satisfy the formula (so there is no need for branching on that literal). The
\emph{pure literal} rule is based on the fact that if a literal occurs in the
formula, but its opposite literal does not, if the formula is satisfiable, in
one of its models that literal is true. These two rules are not necessary for
completeness, although they have a significant impact on efficiency.

\paragraph{Passing valuations instead of modifying the formula.}
In the original \dpll{} procedure, the formula considered is passed as a
function argument, and modified throughout recursive calls. This is unacceptably
inefficient for huge propositional formulae and can be replaced by a procedure
that maintains a current (partial) valuation $M$ and, rather than modifying the
formula, keeps the formula constant and checks its value against the current
valuation (see Figure \ref{fig:moddpll}). The inference rules used in the
original procedure must be adapted to fit this variant of the algorithm.  The
unit clause rule then states that if there is a clause in $F$ such that all its
literals, except exactly one, are false in $M$, and that literal is undefined in
$M$, then this literal must be added to $M$ in order to satisfy this clause.
The pure literal rule turns out to be too expensive in this context, so modern
solvers typically do not use it.

\begin{figure}[ht]
\begin{flushleft}
\begin{footnotesize}
\texttt{%
function {\fn dpll} ($M$ : Valuation) : (SAT, UNSAT)\\
begin\\
\tab if $M \falsifies F$ then return UNSAT\\
\tab else if $M$ is total wrt.~the variables of $F$ then return SAT\\
\tab else if there is a unit clause (i.e., a clause\\
\tab \tab $l \vee l_1 \vee \ldots \vee l_k$ in $F$ s.t.~$l, \opp{l} \notin M$, $\opp{l_1}, \ldots, \opp{l_k} \in M$) then return {\fn dpll}($M \cup \{l\}$)\\
\tab else begin\\
\tab \tab select a literal $l$ s.t.~$l \in F$, $l, \opp{l} \notin M$ \\
\tab \tab if {\fn dpll}($M \cup \{l\}$) = SAT then return SAT\\
\tab \tab else return {\fn dpll}($M \cup \{\opp{l}\}$)\\
\tab end\\
end}
\end{footnotesize}
\end{flushleft}
\caption{DPLL procedure with valuation passing} \label{fig:moddpll}
\end{figure}

\paragraph{Non-recursive implementation.}
To gain efficiency, modern \sat{} solvers implement \dpll-like procedures in a
non-recursive fashion. Instead of passing arguments through recursive calls,
both the current formula $F$ and the current partial valuation $M$ are kept as
global objects. The valuation acts as a stack and is called \emph{assertion
  trail}. Since the trail represents a valuation, it must not contain repeated
nor opposite literals (i.e., it is always {\em distinct} and {\em consistent}).
Literals are added to the stack top ({\em asserting}) or removed from the stack
top ({\em backtracking}). The search begins with an empty trail. During the
solving process, the solver selects literals undefined in the current trail $M$
and asserts them, marking them as \emph{decision literals}. Decision literals
partition the trail into {\em levels}, and the level of a literal is the number
of decision literals that precede that literal in the trail. After each
decision, unit propagation is exhaustively applied and unit literals are
asserted to $M$, but as \emph{implied literals} (since they are not arbitrary
decisions). This process repeats until either (i) a clause in $F$ is found which
is false in the current trail $M$ (this clause is called a \emph{conflict
  clause}) or (ii) all the literals occurring in $F$ are defined in $M$ and no
conflict clause is found in $F$. In the case (i), a conflict reparation
(backtracking) procedure must be applied. In the basic variant of the conflict
reparation procedure, the last decision literal $l$ and all literals after it
are backtracked from $M$, and the opposite literal of $l$ is asserted, also as
an \emph{implied literal}. If there is no decision literal in $M$ when a
conflict is detected, then the formula $F$ is unsatisfiable. In the case (ii),
the formula is found to be satisfiable and $M$ is its model.

\paragraph{Modern \dpll{} enhancements.}
For almost half of a century, DPLL-based \sat{} procedures have undergone
various modifications and improvements. Accounts of the evolution of SAT solvers
can be found in recent literature \cite{sathandbook,satsolvers-hkr}.  Early
\sat{} solvers based on DPLL include Tableau (NTAB), POSIT, 2cl and CSAT, among
others. In the mid 1990's, a new generation of solvers such as GRASP
\cite{grasp}, SATO \cite{sato}, Chaff \cite{chaff}, and BerkMin \cite{berkmin}
appeared, and in these solvers a lot of attention was payed to optimisation of
various aspects of the DPLL algorithm. Some influential modern \sat{} solvers
include MiniSat \cite{minisat} and PicoSAT \cite{picosat}.

A significant improvement over the basic search algorithm is to replace the
simple conflict reparation based on backtracking by a more advanced one based on
\emph{conflict driven backjumping}, first proposed in the Constraint
Satisfaction Problem (CSP) domain \cite{zhang-csp}. Once a conflict is detected,
a \emph{conflict analysis procedure} finds sequence of decisions (often buried
deeper in the trail) that eventually led to the current conflict. Conflict
analysis can be described in terms of graphs and the backjump clauses are
constructed by traversing a graph called {\em implication graph}
\cite{grasp}. The process can also be described in terms of resolution that
starts from the conflict clause and continues with clauses that caused unit
propagation of literals in that clause \cite{zhangmalik-quest}.  There are
several strategies for conflict analysis, leading to different backjump clauses
\cite{sathandbook}.  Most conflict analysis strategies are based on the
following scheme:
\begin{enumerate}[(1)]
\packitems
\item Conflict analysis starts with a conflict clause (i.e., the clause from $F$
  detected to be false in $M$). The conflict analysis clause $C$ is set to the
  conflict clause.
\item Each literal from the current conflict analysis clause $C$ is false in the
  current trail $M$ and is either a decision literal or a result of a
  propagation. For each propagated literal $l$ it is possible to find a clause
  \emph{(reason clause)} that caused $l$ to be propagated. The propagated
  literals from $C$ are then replaced (it will be said {\em explained}) by
  remaining literals from their reason clauses. The process of conflict analysis
  then continues.
\end{enumerate}
The described procedure continues until some termination condition is met, and
the backjump clause is then constructed. Thanks to conflict driven backjumping,
a lot of unnecessary work can be saved compared to the simple backtrack
operation. Indeed, the simple backtracking would have to consider all
combinations of values for all decision literals between the backjump point and
the last decision, while they are all irrelevant for the particular conflict.

The result of conflict analysis is usually a clause that is a logical
consequence of $F$ and that explains a particular conflict that occurred. If
this clause was added to $F$, then this type of conflict would occur never again
during search (even in some other contexts, i.e., in some other parts of the
search space). This is why solvers usually perform \emph{clause learning} and
append (redundant) deduced clauses to $F$.  However, if the formula $F$ becomes
too large, some clauses have to be \emph{forgotten}.  Conflict driven
backjumping with clause learning were first incorporated into a SAT solver in
the mid 1990's by Silva and Sakallah in GRASP \cite{grasp} and by Bayardo and
Schrag in rel\verb|_|sat \cite{relsat}. \dpll-based \sat{} solvers employing
conflict driven clause learning are often called \emph{CDCL solvers}.

Another significant improvement is to empty the trail and \emph{restart} the
search from time to time, in a hope that it would restart in an easier part of
the search space. Randomized restarts were introduced by Gomes et
al.~\cite{gomes-randomization} and further developed by Baptista and
Marques-Silva \cite{baptista-randomization}.

One of the most demanding operations during solving is the detection of false
and unit clauses. Whenever a literal is asserted, the solver must check $F$ for
their presence. To aid this operation, smart data structures with corresponding
implementations are used.  One of the most advanced ones is the
\emph{two-watched literal scheme}, introduced by Moskewicz et al.~in their
solver zChaff \cite{chaff}.

\subsection{Abstract State Transition Systems}

An \emph{abstract state transition system} for an imperative program consists of
a set of \emph{states} $S$ describing possible values of the program's global
variables and a binary transition relation $\rightarrow\ \subseteq S \times
S$. The transition relation is usually the union of smaller transition relations
$\rightarrow_i$, called the \emph{transition rules}. If $s \rightarrow_i s'$
holds, we say that the rule $i$ has been applied to the state $s$ and the state
$s'$ has been obtained. Transition rules are denoted as:

\begin{center}
\begin{eqnarray*}
{\fn Rule name}: & &
\begin{array}{c}
\mathit{cond}_1\quad\ldots\quad \mathit{cond}_k \\ \hline 
\mathit{effect}
\end{array} 
\end{eqnarray*}
\end{center}

\noindent Above the line are the conditions $\mathit{cond}_1$, \ldots,
$\mathit{cond}_k$ that the state $s$ must meet in order for the rule to be
applicable and the $\mathit{effect}$ denotes the effect that must be applied to
the components of $s$ in order to obtain $s'$.

More formally, transition rules can be defined as relations over states:
$${\fn Rule name}\ s\ s' \iff \phi$$
\noindent
where $\phi$ denotes a formula that describes conditions on $s$ that have to be
met and the relationship between $s$ and $s'$.

Some states are distinguished as \emph{initial states}. An initial state usually
depends on the program input. A state is a \emph{final state} if no transition
rules can be applied. Some states (not necessarily final) are distinguished as
the \emph{outcome states} carrying certain resulting information.  If a program
terminates in a final outcome state, it emits a result determined by this
state. For a decision procedure (such as a \sat{} solver), there are only two
possible outcomes: {\em yes} ({\em sat}) or {\em no} ({\em unsat}).  A state
transition system is considered to be correct if it has the following
properties:
\begin{desCription}
\packitems
\item\noindent{\hskip-12 pt\bf Termination:}\ from each initial state $s_0$, the execution eventually
  reaches a final state (i.e., there are no infinite chains $s_0 \rightarrow s_1
  \rightarrow \ldots$).
\item\noindent{\hskip-12 pt\bf Soundness:}\ the program always gives correct answers, i.e., if the
  program, starting with an input $I$ from an initial state $s_0$, reaches a
  final outcome state with a result $O$, then $O$ is the desired result for the
  input $I$.
\item\noindent{\hskip-12 pt\bf Completeness:}\ the program always gives an answer if it terminates,
  i.e., all final states are outcome states.
\end{desCription}

\subsection{Abstract State Transition Systems for \sat{}}

Two transition rule systems that model DPLL-based \sat{} solvers and related SMT
solvers have been published recently. Both systems present a basis of the
formalization described in this paper. The system of Krsti\'c and Goel
\cite{Krstic-Frocos07} gives a more detailed description of some parts of the
solving process (particularly the conflict analysis phase) than the one given by
Nieuwenhuis, Oliveras and Tinelli \cite{NieOT-Jacm06}, so we present its rules
in Figure \ref{fig:krstictransitions}. In this system, along with the formula
$F$ and the trail $M$, the state of the solver is characterized by the conflict
analysis set $C$ that is either a set of literals (i.e., a clause) or the
distinguished symbol $\mathit{no\_cflct}$. Input to the system is an arbitrary
set of clauses $F_0$. The solving starts from a initial state in which $F=F_0$,
$M=[\,]$, and $C=\mathit{no\_cflct}$.

\begin{figure}[ht!]
\begin{flushleft}
\begin{small}
\begin{eqnarray*}
{\fn Decide}: & &  
\begin{array}{c}
l\in L\qquad l,\opp{l}\notin M\\ \hline 
M := M\ \decision{l}
\end{array} \\
{\fn UnitPropag}: & &
\begin{array}{c}
l\vee l_1\vee\ldots\vee l_k \in F \qquad
\opp{l}_1,\ldots,\opp{l}_k\in M \qquad
l, \opp{l}\notin M \\ \hline 
M := M\ \nondecision{l}
\end{array} \\
{\fn Conflict}:  & &
\begin{array}{c}
C = \mathit{no\_cflct}\qquad \opp{l}_1\vee\ldots\vee\opp{l}_k\in F\qquad
l_1,\ldots,l_k\in M \\ \hline 
C := \{l_1,\ldots,l_k\}
\end{array} \\
{\fn Explain}:  & &
\begin{array}{c}
l\in C\qquad l\vee\opp{l}_1\vee\ldots\vee \opp{l}_k\in F\qquad l_1,\ldots,l_k \prec l \\ \hline 
C := C \cup \{l_1,\ldots,l_k\}\setminus\{l\}
\end{array} \\
{\fn Learn}:  & &
\begin{array}{c}
C =\{l_1, \ldots, l_k\}\qquad \opp{l}_1\vee\ldots\vee\opp{l}_k\notin F \\ \hline 
F := F \cup \{\opp{l}_1\vee\ldots\vee \opp{l}_k\}
\end{array} \\
{\fn Backjump}:  & &
\begin{array}{c}
C =\{l, l_1, \ldots, l_k\}\qquad \opp{l}\vee\opp{l}_1\vee\ldots\vee\opp{l}_k\in F\qquad
\textrm{level } l > m \geq \textrm{level } l_i \\ \hline 
C:=\mathit{no\_cflct}\qquad M:=M^{[m]}\ \nondecision{\opp{l}}
\end{array} \\
{\fn Forget}: & &
\begin{array}{c}
C = \mathit{no\_cflct}\qquad c \in F\qquad F\setminus c\vDash c \\ \hline 
F := F \setminus c 
\end{array} \\
{\fn Restart}: & &
\begin{array}{c}
C = \mathit{no\_cflct} \\ \hline 
M := M^{[0]}
\end{array} \\
\end{eqnarray*}

\caption{Transition system for \sat{} solving by Krsti\'c and Goel ($l_i \prec
  l_j$ denotes that the literal $l_i$ precedes $l_j$ in $M$, $\decision{l}$ denotes a
  decision literal, $\nondecision{l}$ an implied literal, $\textrm{level}\ l$ denotes the
  decision level of a literal $l$ in $M$, and $M^{[m]}$ denotes the prefix of
  $M$ up to the level $m$).}
\label{fig:krstictransitions}
\end{small}
\end{flushleft}
\end{figure}

The \rulename{Decide} rule selects a literal from a set of decision literals $L$
and asserts it to the trail as a decision literal. The set $L$ is typically just
the set of all literals occurring in the input formulae. However, in some cases
a smaller set can be used (based on some specific knowledge about the encoding
of the input formula). Also, there are cases when this set is in fact larger
than the set of all variables occurring in the input formula.\footnote{For
  example, the standard DIMACS format for \sat{} requires specifying the number
  of variables and the clauses that make the formula, without guarantees that
  every variable eventually occurs in the formula.}

The \rulename{UnitPropag} rule asserts a unit literal $l$ to the trail $M$ as an
implied literal. This reduces the search space since only one valuation for $l$
is considered.

\noindent The \rulename{Conflict} rule is applied when a conflict clause is detected.  It
initializes the conflict analysis and the reparation procedure, by setting $C$
to the set of literals of the conflict clause. This set is further refined by
successive applications of the \rulename{Explain} rule, which essentially
performs a resolution between the clause $C$ and the clause that is the reason
of propagation of its literal $l$. During the conflict analysis procedure, the
clause $C$ can be added to $F$ by the \rulename{Learn} rule.  However, this is
usually done only once --- when there is exactly one literal in $C$ present at
the highest decision level of $M$. In that case, the \rulename{Backjump} rule
can be applied. That resolves the conflict by backtracking the trail to a level
(usually the lowest possible) such that $C$ becomes unit clause with a unit
literal $l$. In addition, unit propagation of $l$ is performed.

The \rulename{Forget} rule eliminates clauses. Namely, because of the learning
process, the number of clauses in the current formula increases. When it becomes
too large, detecting false and unit clauses becomes too demanding, so from time
to time, it is preferable to delete from $F$ some clauses that are
redundant. Typically, only learnt clauses are forgotten (as they are always
redundant).

\section{Underlying Theory}
\label{sec:backgroundtheory}

As a framework of our formalization, higher-order logic is used, in a similar
way as in the system Isabelle/HOL \cite{isabelle}. Formulae and logical
connectives of this logic ($\wedge$, $\vee$, $\neg$, $\longrightarrow$,
$\longleftrightarrow$) are written in the standard way. Equality is denoted by
$=$. Function applications are written in prefix form, as in ${\fn f}\ x_1\
\ldots\ x_n$. Existential quantification is denoted by $\exists\ x.\ ...$ and
universal quantification by $\forall\ x.\ ...$.

In this section we will introduce definitions necessary for notions of
satisfiability and notions used in \sat{} solving. Most of the definitions are
simple and technical so we give them in a very dense form.  They make the paper
self-contained and can be used just for reference.

The correctness of the whole formalization effort eventually relies on the
definition of satisfiable formulae, which is rather straightforward and easily
checked by human inspection.

\subsection{Lists, Multisets, and Relations}

We assume that the notions of ordered pairs, lists and (finite) sets are defined
within the theory. Relations and their extensions are used primarily in the
context of ordering relations and the proofs of termination. We will use
standard syntax and semantics of these types and their operations. However, to
aid our formalization, some additional operations are introduced.

\begin{defi}[Lists related]\hfill
\label{def:precedesList}
\begin{iteMize}{$\bullet$}
\packitems
\item \emph{The first position of an element $e$ in a list $l$}, denoted
  ${\fn firstPos}\ e\ l$, is the zero-based index of the first occurrence of $e$
  in $l$ if it occurs in $l$ or the length of $l$ otherwise.
\item \emph{The prefix to an element $e$ of a list $l$}, denoted by ${\fn
  prefixTo}\ e\ l$, is the list consisting of all elements of $l$ preceding the
  first occurrence of $e$ (including $e$).
\item \emph{The prefix before an element $e$ of a list $l$}, denoted by ${\fn
  prefixBefore}\ e\ l$ is the list of all elements of $l$ preceding the first
  occurrence of $e$ (not including $e$).
\item \emph{An element $e_1$ precedes $e_2$ in a list $l$}, denoted by
  $e_1\precedes{l} e_2$, if both occur in $l$ and the first position of $e_1$ in
  $l$ is less than the first position of $e_2$ in $l$.
\item A list $p$ is a \emph{prefix} of a list $l$ (denoted by $p \leq l$) if
  there exists a list $s$ such that $l = p\,@\,s$.
\end{iteMize}
\end{defi}

\begin{defi}[Multiset]
  A multiset over a type $X$ is a function $S$ mapping $X$ to natural numbers. A
  multiset is \emph{finite} if the set $\{x\ |\ S(x)>0\}$ is finite. \emph{The
    union} of multisets $S$ and $T$ is a function defined as $(S \cup
  T)(x)=S(x)+T(x)$.
\end{defi}

\begin{defi}[Relations related] \hfill
  \label{def:lexext}
\begin{iteMize}{$\bullet$}
\packitems
\item
 The composition of two relations $\rho_1$ and $\rho_2$ is denoted by $\rho_1
 \circ \rho_2$. The $n$-th degree of the relation $\rho$ is denoted by
 $\rho^n$. The transitive closure of $\rho$ is denoted by $\rho^+$, and the
 transitive and reflexive closure of $\rho$ by $\rho^*$.  
 
\item A relation $\succ$ is \emph{well-founded} iff:
\begin{small}
  $$\forall P.\ ((\forall x.\ (\forall y.\ x \succ y \implies P(y)) \implies
  P(x)) \implies (\forall x.\ P(x)))$$
\end{small}

\item
  If $\succ$ is a relation on $X$, then its \emph{lexicographic extension
    $\lex{\succ}$} is a relation on lists of $X$, defined by:
\begin{small}
  \begin{eqnarray*}
    s \lex{\succ} t &\iff& (\exists\ r.\ s = t\,@\,r\ \wedge\ r\neq [\,])\ \vee\
    \\
    & & (\exists\ r\,s'\,t'\,a\,b.\ s=r\,@\,a\,@\,s'\ \wedge\ t=r\,@b\,@\,t'\
    \wedge\ a \succ b)
  \end{eqnarray*}
\end{small}

\item
  If $\succ$ is a  relation on $X$, then its \emph{multiset extension $\mult{\succ}$} is
  a relation defined over multisets over $X$ (denoted by $\langle x_1, \ldots,
  x_n \rangle$). The relation $\mult{\succ}$ is a transitive closure of the relation
  $\mults{\succ}$, defined by:
  \begin{small}
  $$
  \begin{array}{rcll}
  S_1 \mults{\succ} S_2\ & \iff\ & \exists  S\ S_2'\ s_1. & S_1 = S\cup\langle s_1\rangle\ \wedge\ S_2 = S\cup S_2'\  \wedge\ \\
  & & & \forall\ s_2.\ s_2\in S_2'\implies s_1 \succ s_2
  \end{array}
  $$
  \end{small}

\item
  \label{def:lexProd}
  Let $\succ_x$ and $\succ_y$ be  relations over $X$ and $Y$. Their \emph{lexicographic
    product}, denoted by $\lexprod{\succ_x}{\succ_y}$, is a relation $\succ$ on $X \times Y$
  such that
\begin{small}
  $$(x_1, y_1) \succ (x_2, y_2) \iff x_1 \succ_x x_2\ \vee\ (x_1 = x_2\ \wedge\ y_1 \succ_y y_2)$$ 
\end{small}

\item Let $\succ_x$ be a relation on $X$, and for each $x \in X$ let $\succ_y^x$
  be a relation over $Y$ (i.e., let $\lambda\ x.\ \succ_y^x$ be a function
  mapping $X$ to relations on $Y$). Their \emph{parametrized lexicographic
    product},\footnote{Note that lexicographic product can be regarded as a
    special case of parametrized lexicographic product (where a same $\succ_y$
    is used for each $x\in X$).}  denoted by $\lexprodp{\succ_x}{\succ_y^x}$, is
  a relation $\succ$ on $X \times Y$ such that
\begin{small}
  $$(x_1, y_1) \succ (x_2, y_2) \iff x_1 \succ_x x_2\ \vee\ (x_1 = x_2\ \wedge\ y_1 \succ_y^{x_1} y_2).$$
\end{small}
\end{iteMize}
\end{defi}

\begin{prop}[Properties of well-founded relations] \hfill
  \label{prop:wfInvImage}
  \label{prop:wfmultextension}
  \label{prop:wfLeksikografskaKombinacijaUredjenja}
\begin{iteMize}{$\bullet$}
\packitems
\item A relation $\succ$ is well-founded iff
\begin{small}
  $$\forall\ Q.\ (\exists\ a\in Q)\implies (\exists\ a_{min}\in
  Q.\ (\forall\ a'.\ a_{min} \succ a' \implies a'\notin Q))$$
\end{small}
\item
  Let ${\fn f}$ be a function and $\succ$ a relation such that $x \succ y \implies {\fn
    f} x \succ' {\fn f} y$. If $\succ'$ is well-founded, then so is $\succ$.
\item
  If $\succ$ is well-founded, then so is $\mult{\succ}$.
\item
  Let $\succ_x$ be a well-founded relation on $X$ and for each $x \in X$ let be
  $\succ_y^x$ a well-founded relation. Then $\lexprodp{\succ_x}{\succ_y^x}$ is 
  well-founded.
\end{iteMize}
\end{prop}

\subsection{Logic of CNF formulae}

\begin{defi}[Basic types]
\label{def:basictypes} \ \ \vspace{2mm}
\newline
{\rm
\begin{small}
\begin{tabular}{ll}
\hline
${\fn Variable}$  &  natural number\\
${\fn Literal}$   &  either a positive variable (${\fn Pos}$ $vbl$) or a negative variable (${\fn Neg}$ $vbl$)\\
${\fn Clause}$    &  a list of literals\\
${\fn Formula}$   &  a list of clauses\\
${\fn Valuation}$ &  a list of literals\\
${\fn Trail}$     &  a list of $({\fn Literal}, {\fn bool})$ pairs\\ 
\hline
\end{tabular}
\end{small}
}
\end{defi}

For the sake of readability, we will sometimes omit types and use the following
naming convention: literals (i.e., variables of the type ${\fn Literal}$) are
denoted by $l$ (e.g., $l, l', l_0, l_1, l_2, \ldots$), variables by $vbl$,
clauses by $c$, formulae by $F$, valuations by $v$, and trails by $M$.

Note that, in order to be closer to implementation (and to the standard solver
input format --- DIMACS), clauses and formulae are represented using lists
instead of sets (a more detailed discussion on this issue is given in Section
\ref{sec:discuss}). Although a trail is not a list of literals (but rather a
list of $({\fn Literal}, {\fn bool})$ pairs), for simplicity, we will often
identify it with its list of underlying literals, and we will treat trails as
valuations.  In addition, a trail can be implemented, not only as a list of
$({\fn Literal}, {\fn bool})$ pairs but in some other equivalent way. We abuse
the notation and overload some symbols. For example, the symbol $\in$ denotes
both set membership and list membership, and it is also used to denote that a
literal occurs in a formula. Symbol ${\fn vars}$ is also overloaded and denotes
the set of variables occurring in a clause, in a formula, or in a valuation.

\begin{defi}[Literals and clauses related] \hfill
\label{def:opp-of-lit} 
\begin{iteMize}{$\bullet$}
\packitems
\item {\em The opposite literal of a literal} $l$, denoted by $\opp{l}$, is defined by: 
$\opp{{\fn Pos}\ vbl}= {\fn Neg}\ vbl$, $\opp{{\fn Neg}\ vbl}={\fn Pos}\ vbl$.

\item A formula $F$ {\em contains a literal} $l$ (i.e., a literal $l$ {\em occurs
   in a formula} $F$), denoted by $l\in F$, iff $\exists c.\ c\in F \wedge l\in c$.

\item The {\em set of variables that occur in a clause $c$} is denoted by ${\fn
    vars}\ c$.
  The {\em set of variables that occur in a formula $F$} is denoted by ${\fn vars}\ F$.
  The {\em set of variables that occur in a
    valuation $v$} is denoted by ${\fn vars}\ v$.

\item The {\em resolvent} of clauses $c_1$ and $c_2$ over the literal $l$, denoted\\
  ${\fn resolvent}\ c_1\ c_2\ l$ is the clause $(c_1 \setminus l) @ (c_2
  \setminus \opp{l})$.

\item A clause $c$ is a {\em tautological clause}, denoted by ${\fn
    clauseTautology}\ c$, if it contains both a literal and its opposite (i.e.,
  $\exists\ l.\ l \in c\,\wedge\,\opp{l}\in c$).

\item The {\em conversion of a valuation $v$ to a formula} is the list
  $\valform{v}$ that contains all single literal clauses made of literals from
  $v$.
\end{iteMize}
\end{defi}

\begin{defi}[Semantics] \hfill
\begin{iteMize}{$\bullet$}
\packitems
\item
  A literal $l$ is {\em true in a valuation} $v$, denoted by $v \vDash l$, iff $l
  \in v$. 
  A clause $c$ is {\em true in a valuation} $v$, denoted by $v \vDash c$, iff
  $\exists l.\ l \in c\wedge v \vDash l$. 
  A formula $F$ is {\em true in a valuation} $v$, denoted by $v \vDash F$, iff
  $\forall c.\ c \in F \Rightarrow v \vDash c$.  

\item
  A literal $l$ is {\em false in a valuation} $v$, denoted by $v \falsifies l$,
  iff $\opp{l} \in v$. A clause $c$ is {\em false in a valuation} 
  $v$, denoted by $v \falsifies c$, iff $\forall l.\ l \in c \Rightarrow v \falsifies l$.
  A formula $F$ is {\em false 
  in a valuation} $v$, denoted by $v \falsifies F$, 
  iff $\exists c.\ c \in F \wedge v \falsifies c$.\footnote{Note that the symbol 
  $\falsifies$ is atomic, i.e., $v \falsifies F$ does not correspond to $v \models (\neg F)$, 
  although it would be the case if all propositional formulae (instead of CNF only) 
  were considered.}
  
\item  
  $v \nvDash l$ ( $v \nvDash c$ /   $v \nvDash F$) denotes that $l$ ($c$ / $F$) 
  is not true in $v$ (then we say that $l$ ($c$~/~$F$) is {\em unsatisfied} in $v$).
  $v \nfalsifies l$ ($v \nfalsifies c$ / $v \nfalsifies F$) denotes that $l$ ($c$ / $F$) 
  is not false in $v$ (then we say that $l$ ($c$ / $F$) is {\em unfalsified} in $v$).
\end{iteMize}
\end{defi}

\begin{defi}[Valuations and models] \hfill
\begin{iteMize}{$\bullet$}
\packitems
\item
  A valuation $v$ is {\em inconsistent}, denoted by ${\fn inconsistent}\ v$, iff it
  contains both a literal and its opposite i.e., iff $\exists l.\ v \vDash l
  \wedge v \vDash \opp{l}$. A valuation is {\em consistent}, denoted by $({\fn
    consistent}\ v)$, iff it is not inconsistent.
\item
  A valuation $v$ is \emph{total} with respect to a variable set $Vbl$, denoted
  by ${\fn total}\ v\ Vbl$, iff ${\fn vars}\ v \supseteq Vbl$.
\item
  A {\em model} of a formula $F$ is a consistent valuation under which $F$ is
  true.  A formula $F$ is {\em satisfiable}, denoted by ${\fn sat}\ F$, iff it
  has a model, i.e., $\exists v.\ {\fn consistent}\ v \wedge v \vDash F$.
\item
  A clause $c$ is {\em unit} in a valuation $v$ with a {\em unit literal} $l$,
  denoted by ${\fn isUnit}\ c\ l\ v$ iff $l\in c$, $v \nvDash l$, $v \nfalsifies l$
  and $v\falsifies (c\setminus l)$ (i.e., $\forall l'.\ l' \in c \wedge l'\neq l
  \Rightarrow v\falsifies l'$).
\item
  A clause $c$ is a {\em reason for propagation} of literal $l$ in valuation
  $v$, denoted by ${\fn isReason}\ c\ l\ v$ iff $l\in c$, $v\vDash l$, $v\falsifies
  (c\setminus l)$, and for each literal $l'\in (c\setminus l)$, the literal
  $\opp{l'}$ precedes $l$ in $v$.
\end{iteMize}
\end{defi}

\newpage
\begin{defi}[Entailment and logical equivalence] \hfill
\begin{iteMize}{$\bullet$} 
\packitems
\item A {\em formula $F$ entails a clause $c$}, denoted by $F\impliesClause c$,
  iff $c$ is true in every model of $F$. A {\em formula $F$ entails a literal
    $l$}, denoted by $F\impliesLiteral l$, iff $l$ is true in every model of
  $F$. A {\em formula $F$ entails valuation $v$}, denoted by $F\impliesValuation
  v$, iff it entails all its literals i.e., $\forall l.\ l\in v\Rightarrow
  F\impliesLiteral l$. A {\em formula $F_1$ entails a formula $F_2$}, denoted by
  $F_1\impliesFormula F_2$, if every model of $F_1$ is a model of $F_2$.
\item
  Formulae $F_1$ and $F_2$ are {\em logically equivalent}, denoted by $F_1\equiv
  F_2$, iff any model of $F_1$ is a model of $F_2$ and vice versa, i.e., iff
  $F_1\impliesFormula F_2$ and $F_2\impliesFormula F_1$.
\end{iteMize}
\end{defi}

\begin{defi}[Trails related] \hfill
\begin{iteMize}{$\bullet$}
\packitems
\item For a trail element $a$, ${\fn element}\ a$ denotes the first (Literal)
  component and ${\fn isDecision}\ a$ denotes the second (Boolean)
  component. For a trail $M$, ${\fn elements}\ M$ 
  denotes the list of all its elements and ${\fn decisions}\ M$ denotes the list
  of all its marked elements (i.e., of all its decision literals).

\item \emph{The last decision literal}, denoted by ${\fn lastDecision}\ M$, is the
  last marked element of the list $M$, i.e., ${\fn lastDecision}\ M = {\fn
    last}\ ({\fn decisions}\ M).$

\item ${\fn decisionsTo}\ M\ l$ is the list of all marked elements from a trail $M$
  that precede the first occurrence of the element $l$, including $l$ if it is
  marked, i.e., ${\fn decisionsTo}\ l\ M = {\fn decisions}\ ({\fn prefixTo}\ l\ M)$.

\item The {\em current level} for a trail $M$, denoted by ${\fn currentLevel}\ M$, is
  the number of marked literals in $M$, i.e., ${\fn currentLevel}\ M = {\fn
    length}\ ({\fn decisions}\ M)$.

\item
  The {\em decision level} of a literal $l$ in a trail $M$, denoted by ${\fn
    level}\ l\ M$, is the number of marked literals in the trail that precede
  the first occurrence of $l$, including $l$ if it is marked, i.e., ${\fn
    level}\ l\ M = {\fn length}\ ({\fn decisionsTo}\ M\ l)$.

\item
  ${\fn prefixToLevel}\ M\ level$ is the prefix of a trail $M$ containing all
  elements of $M$ with levels less or equal to $level$.

\item
  The \emph{prefix before last decision}, denoted by ${\fn
    prefixBeforeLastDecision}\ M$, is a prefix of the trail $M$ before its last
  marked element (not including it),\footnote{Note that some of these functions
  are used only for some trails. For example, ${\fn prefixBeforeLastDecision}\ M$ 
  makes sense only for trails that contain at least one decision literal.
  Nevertheless, these functions are still defined as total functions --- for example,
  ${\fn prefixBeforeLastDecision}\ M$ equals $M$ if there are no decision literals.}

\item
  The {\em last asserted literal of a clause} $c$, denoted by ${\fn
    lastAssertedLiteral}\ c\ {\elements M}$, is the literal from $c$ that is in
  $\elements M$, such that no other literal from $c$ comes after it in
  $\elements M$.

\item
  The \emph{maximal level of a literal in the clause} $c$ with respect to a
  trail $M$, denoted by ${\fn maxLevel}\ c\ M$, is the maximum of all levels of
  literals from $c$ asserted in $M$.
\end{iteMize}
\end{defi}

\begin{exa}
\label{ex:1}A trail $M$ could be $[\nondecision{+1},
\decision{-2}, \nondecision{+6}, \decision{+5}, \nondecision{-3},
\nondecision{+4}, \decision{-7}]$. The symbol $+$ is written instead of the
constructor ${\fn Pos}$, the symbol $-$ instead of ${\fn Neg}$.  ${\fn
  decisions}\ M $ $=$ $ [\decision{-2}, \decision{+5}, \decision{-7}]$, ${\fn
  lastDecision}\ M$ $=$ $-7,$ ${\fn decisionsTo}\ M\ \mbox{+4}$ $=$
$[\decision{-2}, \decision{+5}]$, and ${\fn decisionsTo}\ M\ \mbox{-7}$ $=$
$[\decision{-2}, \decision{+5}, \decision{-7}]$.  ${\fn level}\ \mbox{+1}\ M =
0$, ${\fn level}\ \mbox{+4}\ M = 2$, ${\fn level}\ \mbox{-7}\ M = 3$, ${\fn
  currentLevel}\ M = 3$, ${\fn prefixToLevel}\ M\ 1 = [\nondecision{+1},
  \decision{+2}, \nondecision{+6}]$.  If $c$ is $[+4, +6,-3]$, then ${\fn
  lastAssertedLiteral}\ c\ {\elements M} = +4,$ and ${\fn maxLevel}\ c\ M$ $=$
$2$.
\end{exa}

\section{DPLL Search}
\label{sec:DPLLSearch}

In this section we consider a basic transition system that contains only
transition rules corresponding to steps used in the original \dpll{} procedure:
unit propagation, backtracking, and making decisions for branching (described
informally in Section \ref{subsec:backround_solving}; the pure literal step is
usually not used within modern \sat{} solvers, so it will be omitted). These
rules will be defined in the form of relations over states, in terms of the
logic described in Section \ref{sec:backgroundtheory}. It will be proved that
the system containing these rules is terminating, sound and complete. The rules
within the system are not ordered and the system is sound, terminating, and
complete regardless of any specific ordering. However, it will be obvious that
better performance is obtained if making decisions is maximally postponed, in
the hope that it will not be necessary.

\subsection{States and Rules}
\label{subsec:DPLLformalization}

The state of the solver performing the basic DPLL search consists of the formula
$F$ being tested for satisfiability (that remains unchanged) and the trail $M$
(that may change during the solver's operation). The only parameter to the
solver is the set of variables $\mathit{DecVars}$ used for branching. By
$\mathit{Vars}$ we will denote the set of all variables encountered during
solving --- these are the variables from the initial formula $F_0$ and the
decision variables $\mathit{DecVars}$, i.e., $\mathit{Vars} = {\fn vars}\
F_0\,\cup\,\mathit{DecVars}$.

\begin{defi}[State]
\label{def:state}
A state of the system is a pair $(M, F)$, where $M$ is a trail and $F$ is a
formula. A state $([\,], F_0)$ is an {\em initial state} for the input formula
$F_0$.
\end{defi}

Transition rules are introduced by the following definition, in the form of 
relations over states.

\begin{small}
\begin{defi}[Transition rules] \ \\
\label{def:dpll-system}
$$
\hspace{-1.1cm}\begin{array}{rcll}
{\fn unitPropagate}\ (M_1, F_1)\ (M_2, F_2) & \iff &
\exists c\ l. & c\in F_1\ \wedge\ {\fn isUnit}\ c\ l\ M_1\ \wedge\ \\
& & & M_2 = M_1\,@\,\nondecision{l} \ \wedge\ F_2 = F_1
\end{array}
$$

$$
\hspace{-2.4cm}{\fn backtrack}\ (M_1, F_1)\ (M_2, F_2) \iff M_1\falsifies
F_1\ \wedge\ {\fn decisions}\ M_1 \neq [\,]\ \wedge
$$
$$\qquad M_2 = {\fn prefixBeforeLastDecision}\ M_1\ @\ \nondecision{\opp{{\fn
    lastDecision}\ M_1}}\ \wedge\ F_2 = F_1
$$

$$
\hspace{-0.5cm}\begin{array}{rcll} {\fn decide}\ (M_1, F_1)\ (M_2, F_2) &\iff&
  \exists l. &{\fn var}\ l\in \mathit{DecVars}\ \wedge\ l\notin
  M_1\ \wedge\ \opp{l}\notin M_1\ \wedge\\ & & & M_2 = M_1\,@\,\decision{l}
  \ \wedge\ F_2 = F_1
\end{array}
$$
\end{defi}
\end{small}

\noindent As can be seen from the above definition (and in accordance with the description
given in Section \ref{subsec:backround_solving}), the rule
\rulename{unitPropagate} uses a {\em unit clause} --- a clause with only one
literal $l$ undefined in $M_1$ and with all other literals false in $M_1$. Such
a clause can be true only if $l$ is true, so this rule extends $M_1$ by $l$ (as
an implied literal). The rule \rulename{backtrack} is applied when $F_1$ is
false in $M_1$. Then it is said that a {\em conflict} occurred, and clauses from
$F_1$ that are false in $M_1$ are called {\em conflict clauses}. In that case,
the last decision literal $\decision{l}$ in $M_1$ and all literals that succeed
it are removed from $M_1$, and the obtained prefix is extended by
$\nondecision{\opp{l}}$ as an implied literal. The rule \rulename{decide}
extends the trail by an arbitrary literal $l$ as a decision literal, such that
the variable of $l$ belongs to $\mathit{DecVars}$ and neither $l$ nor $\opp{l}$
occur in $M_1$. In that case, we say there is a {\em branching} on $l$.

The transition system considered is described by the relation $\rightarrowd$,
introduced by the following definition.

\begin{defi}[$\rightarrowd$]
\label{def:rightarrow}
$$s_1 \rightarrowd s_2 \iff {\fn unitPropagate}\ s_1\ s_2\ \vee\ {\fn
  backtrack}\ s_1\ s_2\ \vee\ {\fn decide}\ s_1\ s_2$$
\end{defi}

\begin{defi}[Outcome states]
An outcome state is either an {\em accepting state} or a {\em rejecting state}.

A state is an {\em accepting state} if $M\nfalsifies F$ and there
is no state $(M',F')$ such that ${\fn decide}\ (M, F)\ (M', F')$ (i.e., there is
no literal such that ${\fn var}\ l\in \mathit{DecVars}$, $l\notin M$, and
$\opp{l}\notin M$).

A state is a {\em rejecting state} if $M\falsifies F$ and ${\fn
  decisions}\ M=[\,]$.
\end{defi}

Note that the condition $M\nfalsifies F$ in the above definition can be replaced
by the condition $M\models F$, but the former is used since its check can be
more efficiently implemented.

\medskip
\begin{exa}
\label{ex:dpll}
Let $F_0$ = $[$ $[-1, +2]$, $[-1, -3, +5, +7]$, $[-1, -2, +5,-7]$, $[-2, +3]$,
  $[+2, +4]$, $[-2, -5, +7]$, $[-3,-6,-7]$, $[-5, +6]$ $]$. One possible
$\rightarrowd$ trace is given below.

\noindent
\begin{tabular}{l|l}
rule                                                        & $M$                                                                                                                        \\ \hline
                                                            & $[\,]$                                                                                                                     \\ 
\rulename{decide} ($l = +1$),                               & $[\decision{+1}]$                                                                                                          \\
\rulename{unitPropagate} ($c = [-1, +2]$, $l = +2$)         & $[\decision{+1}, \nondecision{+2}]$                                                                                        \\
\rulename{unitPropagate} ($c = [-2, +3]$, $l = +3$)         & $[\decision{+1}, \nondecision{+2}, \nondecision{+3}]$                                                                      \\
\rulename{decide} ($l = +4$)                                & $[\decision{+1}, \nondecision{+2}, \nondecision{+3}, \decision{+4}]$                                                       \\
\rulename{decide} ($l = +5$)                                & $[\decision{+1}, \nondecision{+2}, \nondecision{+3}, \decision{+4}, \decision{+5}]$                                        \\
\rulename{unitPropagate} ($c = [-5, +6]$, $l = +6$)         & $[\decision{+1}, \nondecision{+2}, \nondecision{+3}, \decision{+4}, \decision{+5}, \nondecision{+6}]$                      \\
\rulename{unitPropagate} ($c = [-2, -5, +7]$, $l = +7$)     & $[\decision{+1}, \nondecision{+2}, \nondecision{+3}, \decision{+4}, \decision{+5}, \nondecision{+6}, \nondecision{+7}]$    \\
\rulename{backtrack} ($M \falsifies [-3, -6, -7]$)          & $[\decision{+1}, \nondecision{+2}, \nondecision{+3}, \decision{+4}, \nondecision{-5}]$                                     \\
\rulename{unitPropagate} ($c = [-1, -3, +5, +7]$, $l = +7$) & $[\decision{+1}, \nondecision{+2}, \nondecision{+3}, \decision{+4}, \nondecision{-5}, \nondecision{+7}]$                   \\
\rulename{backtrack} ($M \falsifies [-1, -2, +5, -7]$)      & $[\decision{+1}, \nondecision{+2}, \nondecision{+3}, \nondecision{-4}]$                                                    \\
\rulename{decide} ($l = +5$)                                & $[\decision{+1}, \nondecision{+2}, \nondecision{+3}, \nondecision{-4}, \decision{+5}]$                                     \\
\rulename{unitPropagate} ($c = [-5, +6]$, $l = +6$)         & $[\decision{+1}, \nondecision{+2}, \nondecision{+3}, \nondecision{-4}, \decision{+5}, \nondecision{+6}]$                   \\
\rulename{unitPropagate} ($c = [-2, -5, +7]$, $l = +7$)     & $[\decision{+1}, \nondecision{+2}, \nondecision{+3}, \nondecision{-4}, \decision{+5}, \nondecision{+6}, \nondecision{+7}]$ \\
\rulename{backtrack} ($M \falsifies [-3, -6, -7]$)          & $[\decision{+1}, \nondecision{+2}, \nondecision{+3}, \nondecision{-4}, \nondecision{-5}]$                                  \\
\rulename{unitPropagate} ($c = [-1, -3, +5, +7]$, $l = +7$) & $[\decision{+1}, \nondecision{+2}, \nondecision{+3}, \nondecision{-4}, \nondecision{-5}, \nondecision{+7}]$                \\
\rulename{backtrack} ($M \falsifies [-1, -2, +5, -7]$       & $[\nondecision{-1}]$                                                                                                       \\
\rulename{decide} ($l = +2$)                                & $[\nondecision{-1}, \decision{+2}]$                                                                                        \\
\rulename{unitPropagate} ($c = [-2, +3]$, $l = +3$)         & $[\nondecision{-1}, \decision{+2}, \nondecision{+3}]$                                                                      \\
\rulename{decide} ($l = +4$)                                & $[\nondecision{-1}, \decision{+2}, \nondecision{+3}, \decision{+4}]$                                                       \\
\rulename{decide} ($l = +5$)                                & $[\nondecision{-1}, \decision{+2}, \nondecision{+3}, \decision{+4}, \decision{+5}]$                                        \\
\rulename{unitPropagate} ($c = [-5, +6]$, $l = +6$)         & $[\nondecision{-1}, \decision{+2}, \nondecision{+3}, \decision{+4}, \decision{+5}, \nondecision{+6}]$                      \\
\rulename{unitPropagate} ($c = [-2, -5, +7]$, $l = +7$)     & $[\nondecision{-1}, \decision{+2}, \nondecision{+3}, \decision{+4}, \decision{+5}, \nondecision{+6}, \nondecision{+7}]$    \\
\rulename{backtrack} $M \falsifies [-3, -6, -7]$            & $[\nondecision{-1}, \decision{+2}, \nondecision{+3}, \decision{+4}, \nondecision{-5}]$                                     \\
\rulename{decide} ($l = +6$)                                & $[\nondecision{-1}, \decision{+2}, \nondecision{+3}, \decision{+4}, \nondecision{-5}, \decision{+6}]$                      \\
\rulename{unitPropagate} ($c = [-3, -6, -7]$, $l = -7$)     & $[\nondecision{-1}, \decision{+2}, \nondecision{+3}, \decision{+4}, \nondecision{-5}, \decision{+6}, \nondecision{-7}]$  \\
\end{tabular}
\end{exa}

\subsection{Properties}
\label{subsec:DPLLProperties}

In order to prove that the presented transition system is terminating, sound,
and complete, first, local properties of the transition rules have to be given
in the form of certain invariants.

\subsubsection{Invariants}
\label{subsec:DPLLInvariants}

For proving properties of the described transition system, several relevant rule
invariants will be used (not all of them are used for proving each of soundness,
completeness, and termination, but we list them all here for the sake of
simplicity).

\medskip

\begin{tabular}{l l}
\defineInvariant{inv:consistent} {\inv{consistent}}:
        & ${\fn consistent}\ M$   \\
\defineInvariant{inv:distinct}     {\inv{distinct}}:
        & ${\fn distinct}\ M$       \\
\defineInvariant{inv:varsM} {\inv{varsM}}:
        & ${\fn vars}\ M\ \subseteq\ \mathit{Vars}$ \\
\defineInvariant{inv:impliedLits} {\inv{impliedLits}}:
        & $\forall l.\ l
           \in M \implies (F\ @\ {\fn decisionsTo}\ l\ M)
           \entailsLiteral l$ \\
\defineInvariant{inv:equiv}      {\inv{equiv}}:
        & $F\equiv F_0$             \\
\defineInvariant{inv:varsF} {\inv{varsF}}:
        & ${\fn vars}\ F\ \subseteq\ \mathit{Vars}$ \\
\end{tabular}

\medskip

The condition \inv{consistent} states that the trail $M$ can potentially be a
model of the formula, and \inv{distinct} requires that it contains no repeating
elements. The \inv{impliedLits} ensures that any literal $l$ in $M$ is entailed
by $F$ with all decision literals that precede $l$.

Notice that the given rules do not change formulae in states, so it trivially
holds that $F=F_0$, which further implies \inv{equiv} and \inv{varsF}. However,
the transition systems that follow in the next sections may change formulae, so
the above set of invariants is more appropriate.  If only testing satisfiability
is considered (and not in building models for satisfiable formulae), instead of
\inv{equiv}, it is sufficient to require that $F$ and $F_0$ are weakly
equivalent (i.e., equisatisfiable).

The above conditions are indeed invariants (i.e., they are met for each state
during the application of the rules), as stated by the following lemma.

\begin{lem}
\label{lemma:invariantsHold}\hfill

\begin{enumerate}[\em(1)]
\packitems
\item In the initial state $([\,], F_0)$ all the invariants hold.
\item If $(M,F)\rightarrowd (M', F')$ and if the invariants are met in the state
  $(M,F)$, then they are met in the state $(M', F')$ too.
\item If $([\,], F_0)\rightarrowd^* (M, F)$, then all the invariants hold in the
  state $(M, F)$.
\end{enumerate}
\end{lem}

The proof of this lemma considers a number of cases --- one for each rule-invariant 
pair.

\subsubsection{Soundness}
\label{subsec:DPLLsoundness}

Soundness of the given transition system requires that if the system terminates
in an accepting state, then the input formula is satisfiable, and if the system
terminates in a rejecting state, then the input formula is unsatisfiable.

The following lemma ensures soundness for satisfiable input formulae, and the
next one is used for proving soundness for unsatisfiable input formulae (but
also in some other contexts).

\begin{lem}
\label{lemma:satReport}
If $\mathit{DecVars} \supseteq {\fn vars}\ F_0$ and if there is 
an accepting state $(M,F)$ such that:
\begin{enumerate}[\em(1)]
\packitems
\item ${\fn consistent}\ M$ (i.e., \inv{consistent} holds),
\item $F\equiv F_0$ (i.e., \inv{equiv} holds),
\item ${\fn vars}\ F\subseteq \mathit{Vars}$ (i.e., \inv{varsF} holds),
\end{enumerate}
then the formula $F_0$ is satisfiable and $M$ is one model (i.e., ${\fn
  model}\ M\ F_0$).
\end{lem}

\begin{lem}
\label{lemma:InvariantImpliedLiteralsAndFormulaFalse}
If there is a state $(M,F)$ such that:
\begin{enumerate}[\em(1)]
\packitems
\item $\forall l.\ l \in M$ $\implies$ $(F\ @\ {\fn decisionsTo}\ l\ M)$
  $\entailsLiteral$ $l$ (i.e., \inv{impliedLits} holds),
\item $M\falsifies F$
\end{enumerate}
then $\neg ({\fn sat}\ (F\ @\ {\fn decisions}\ M))$.
\end{lem}

\begin{thm}[Soundness for $\rightarrowd$]
\label{thm:soundness}
If $([\,], F_0) \rightarrowd^* (M, F)$, then:
\begin{enumerate}[\em(1)]
\packitems
\item If $\mathit{DecVars} \supseteq {\fn vars}\ F_0$ and $(M,F)$ is an
  accepting state, then the formula $F_0$ satisfiable and $M$ is one model
  (i.e., ${\fn sat}\ F_0$ and ${\fn model}\ M\ F_0$).
\item If $(M,F)$ is a rejecting state, then the formula $F_0$ is unsatisfiable
  (i.e., $\neg({\fn sat}\ F_0)$).
\end{enumerate}
\end{thm}

\begin{proof}
By Lemma \ref{lemma:invariantsHold} all the invariants hold in the state $(M,
F)$.

Let us assume that $\mathit{DecVars} \supseteq {\fn vars}\ F_0$ and $(M,F)$ is
an accepting state. Then, by Lemma \ref{lemma:satReport}, the formula is $F_0$
satisfiable and $M$ is one model.

Let us assume that $(M,F)$ is a rejecting state. Then $M\falsifies F$ and, by
Lemma \ref{lemma:InvariantImpliedLiteralsAndFormulaFalse}, 
$\neg ({\fn sat}\ (F\ @\ ({\fn decisions}\ M)))$.  Since $(M,F)$ is a rejecting 
state, it holds that ${\fn decisions}\ M=[\,]$, and hence $\neg ({\fn sat}\ F)$. From 
$F\equiv F_0$ (\inv{equiv}), it follows that $\neg ({\fn sat}\ F_0)$, i.e., the formula
$F_0$ is unsatisfiable.
\end{proof}

\subsubsection{Termination}
\label{sec:DPLLtermination}

Full and precise formalization of termination is very demanding, and termination
proofs given in the literature (e.g., \cite{Krstic-Frocos07,NieOT-Jacm06}) are
far from detailed formal proofs. For this reason, termination proofs will be
presented here in more details, including auxiliary lemmas used to prove the
termination theorem.

The described transition system terminates, i.e., for any input formulae $F_0$,
the system (starting from the initial state $([\,], F_0)$) will reach a final
state in a finite number of steps. In other words, the relation $\rightarrowd$
is well-founded. This can be proved by constructing a well-founded partial
ordering $\succ$ over trails, such that $(M_1,F_1)\rightarrowd (M_2, F_2)$
implies $M_1 \succ M_2$. In order to reach this goal, several auxiliary
orderings are defined.

First, a partial ordering over annotated literals $\succlit$ and a partial
ordering over trails $\succTr$ will be introduced and some of their properties
will be given within the following lemmas.

\begin{defi}[$\succlit$]
\label{def:l-ordering}
$l_1 \succlit l_2 \iff {\fn isDecision}\ l_1 \wedge \neg({\fn isDecision}\ l_2)$
\end{defi}

\begin{lem}
\label{lemma:preclittrans}
$\succlit$ is transitive and irreflexive.
\end{lem}

\begin{defi}[$\succTr$]
\label{def:M-ordering}
$M_1\succTr M_2 \iff M_1 \lex{\succlit} M_2,$ where $\lex{\succlit}$ is a
lexicographic extension of $\succlit$.
\end{defi}

\begin{lem}
\label{lemma:succtr_properties}
$\succTr$ is transitive, irreflexive, and acyclic (i.e., there is no trail $M$
such that $M\,\transclosure{\succTr}\,M$).

For any three trails $M$, $M'$, and $M''$ it holds that: if $M' \succTr M''$,
then $M\ @\ M' \succTr M\ @\ M''$.
\end{lem}

The next lemma links relations $\rightarrowd$ and $\succTr$.

\begin{lem}\label{lemma:rightarrowsucc}
If ${\fn decide}\ (M_1, F_1)\ (M_2, F_2)$ or ${\fn unitPropagate}\ (M_1,
F_1)\ (M_2, F_2)$ or ${\fn backtrack}$ $(M_1, F_1)\ (M_2, F_2)$, then $ M_1
\succTr M_2$.
\end{lem}

The relation $\succTr$ is not necessarily well-founded (for the elements of the
trails range over infinite sets), so a restriction $\succTrrestrict$ of the
relation $\succTr$ will be defined such that it is well-founded, which will lead
to the termination proof for the system.

\begin{defi}[$\succTrrestrict$]
$
M_1 \succTrrestrict M_2 \iff  ({\fn distinct}\ M_1\ \wedge\ {\fn vars}\ M_1\subseteq Vbl)\ \wedge\ 
                             ({\fn distinct}\ M_2\ \wedge\ {\fn vars}\ M_2\subseteq Vbl)\ \wedge\ 
                              M_1 \succTr M_2
$
\end{defi}

\begin{lem}
\label{lemma:wfTrailSuccRestricted}
If the set $Vbl$ is finite, then the relation $\succTrrestrict$ is a
well-founded ordering.
\end{lem}

Finally, we prove that the transition system is terminating.

\begin{thm}[Termination for $\rightarrowd$]
\label{thm:termination}
  If the set $\mathit{DecVars}$ is finite, for any formula $F_0$, the relation
  $\rightarrowd$ is well-founded on the set of states $(M,F)$ such that $([\,],
  F_0)\rightarrowd^* (M,F)$.
\end{thm}

\begin{proof}
  By Proposition \ref{prop:wfInvImage} it suffices to construct a well-founded
  ordering on the set of states $(M,F)$ such that $([\,], F_0)\rightarrowd^*
  (M,F)$ such that
  $$(M_1, F_1) \rightarrowd (M_2, F_2) \implies (M_1, F_1) \succ (M_2, F_2).$$
  One such ordering is $\succ$ defined by:
  $(M_1, F_1) \succ (M_2, F_2) \iff M_1 \succTrrestrict[\mathit{Vars}] M_2.$

  Indeed, since by Lemma \ref{lemma:wfTrailSuccRestricted},
  $\succTrrestrict[\mathit{Vars}]$ is well-founded, by Proposition
  \ref{prop:wfInvImage} (for a function mapping $(M, F)$ to $M$), $\succ$ is
  also a well-founded ordering.

  Let $(M_1, F_1)$ and $(M_2, F_2)$ be two states such that $([\,],
  F_0)\rightarrowd^*(M_1, F_1)$ and $(M_1, F_1) \rightarrowd (M_2, F_2)$. By
  Lemma \ref{lemma:invariantsHold} all the invariants hold for $(M_1,
  F_1)$. From $(M_1, F_1) \rightarrowd (M_2, F_2)$, by Lemma
  \ref{lemma:rightarrowsucc}, it follows that $M_1\succTr M_2$. Moreover, by
  Lemma \ref{lemma:invariantsHold}, all the invariants hold also for $(M_2,
  F_2)$, so ${\fn distinct}\ M_1$, ${\fn vars}\ M_1 \subseteq
  \mathit{Vars}$, ${\fn distinct}\ M_2$ and ${\fn vars}\ M_2 \subseteq
  \mathit{Vars}$. Ultimately, $M_1\succTrrestrict[\mathit{Vars}]
  M_2$.
\end{proof}

\subsubsection{Completeness}
\label{sec:DPLLcompleteness}

Completeness requires that all final states are outcome states.

\begin{thm}[Completeness for $\rightarrowd$]
\label{thm:completeness}
Each final state is either accepting or rejecting.
\end{thm}

\begin{proof}
Let $(M, F)$ be a final state. It holds that either $M\falsifies F$ or
$M\nfalsifies F$.

If $M\nfalsifies F$, since there is no state $(M',F')$ such that ${\fn
  decide}\ (M, F)\ (M', F')$ (as $(M,F)$ is a final state), there is no
literal $l$ such that ${\fn var}\ l\in \mathit{DecVars}$, $l\notin M$, and
$\opp{l}\notin M$, so $(M,F)$ is an accepting state.

If $M\falsifies F$, since there is no state $(M',F')$ such that ${\fn
  backtrack}\ (M, F)\ (M', F')$ (as $(M,F)$ is a final state), it holds
that ${\fn decisions}\ M = [\,]$, so $(M,F)$ is a rejecting state.
\end{proof}

Notice that from the proof it is clear that the basic search system consisting
only of the rules \rulename{decide} and \rulename{backtrack} is complete.

\subsubsection{Correctness}
\label{sec:DPLLcorrectness}

The theorems \ref{thm:soundness}, \ref{thm:termination}, and
\ref{thm:completeness} directly lead to the theorem about correctness of the
introduced transition system.\footnote{Correctness of the system can be proved
  with a weaker condition.  Namely, instead of the condition that all variables
  of the input formula belong to the set $\mathit{DecVars}$, it is sufficient
  that all {\em strong backdoor} variables belong to $\mathit{DecVars}$
  \cite{sathandbook}, but that weaker condition is not considered here.}

\begin{thm}[Correctness for $\rightarrowd$]
\label{thm:correctnes}
The given transition system is correct, i.e., if all variables of the input
formula belong to the set $\mathit{DecVars}$, then for any satisfiable input
formula, the system terminates in an accepting state, and for any unsatisfiable
formula, the system terminates in a rejecting state.
\end{thm}

\section{Backjumping}
\label{sec:backjumping} 

In this section, we consider a transition system that replaces naive
chronological backtracking by more advanced nonchronological backjumping.

\subsection{States and Rules}
\label{subsec:backjump_rules}

The rules of the new system are given (as in Section \ref{sec:DPLLSearch}) in
the form of relations over states.

\begin{small}
\begin{defi}[Transition rules]  
\label{def:backjumping_system}
\ \\
${\fn unitPropagate}\ (M_1, F_1)\ (M_2, F_2) \iff$
\vspace{-1mm}
\begin{eqnarray*}
\exists c\ l.& &F_1\entailsClause c \ \wedge\ {\fn var}\ l \in \mathit{Vars} \ \wedge\ {\fn isUnit}\ c\ l\ M_1\ \wedge\ \\
& & M_2 = M_1\,@\,\nondecision{l} \ \wedge\ F_2 = F_1
\end{eqnarray*} \ \\
${\fn backjump}\ (M_1, F_1)\ (M_2, F_2) \iff$
\vspace{-1mm}
\begin{eqnarray*}
\exists\ c\ l\ P\ \mathit{level}.& &  F_1\entailsClause c \ \wedge\ {\fn var}\ l\in \mathit{Vars}\ \wedge\ \\
                     & & P = {\fn prefixToLevel}\ \mathit{level}\ M \ \wedge\ 0\leq \mathit{level} < {\fn currentLevel}\ M \ \wedge\ \\
                     & & {\fn isUnit}\ c\ l\ P\ \wedge\ \\
                     & & F_2 = F_1 \ \wedge\ M_2 = P\,@\,\nondecision{l}
\end{eqnarray*}

$
\hspace{-0.65cm}\begin{array}{rcll}
{\fn decide}\ (M_1, F_1)\ (M_2, F_2) &\iff& \exists l. &{\fn var}\ l\in \mathit{DecVars}\ \wedge\ l\notin M_1\ \wedge\ \opp{l}\notin M_1\ \wedge\\
& & & M_2 = M_1\,@\,\decision{l} \ \wedge\ F_2 = F_1
\end{array}
$
\end{defi}
\end{small}

In the following, the transition system described by the relation $\rightarrowb$
defined by these rules will be considered. 

The key difference between the new transition system and one built over the
rules given in Definition \ref{def:dpll-system} is the rule \rulename{backjump}
(that replaces the rule \rulename{backtrack}). The rule \rulename{decide} is
the same as the one given in Definition \ref{def:dpll-system}, while the rule
\rulename{unitPropagate} is slightly modified (i.e., its guard is relaxed).

The clause $c$ in the \rulename{backjump} rule is called a {\em backjump clause}
and the level $level$ is called a {\em backjump level}. The given definition of
the \rulename{backjump} rule is very general --- it does not specify how the
backjump clause $c$ is constructed and what prefix $P$ (i.e., the level $level$)
is chosen if there are several options. There are different strategies that
specify these choices and they are required for concrete implementations. The
conditions that $P$ is a prefix to a level (i.e., that $P$ is followed by a
decision literal in $M_1$) and that this level is smaller than the current level
are important only for termination. Soundness can be proved even with a weaker
assumption that $P$ is an arbitrary prefix of $M_1$. However, usually the
shortest possible prefix $P$ is taken.  The \rulename{backtrack} rule can be
seen as a special case of the \rulename{backjump} rule. In that special case,
the clause $c$ is built of opposites of all decision literals in the trail and
$P$ becomes ${\fn prefixBeforeLastDecision} \ M_1$.

Notice that the backjump clause $c$ does not necessarily belong to $F_1$ but
can be an arbitrary logical consequence of it. So, instead of $c\in F_1$, weaker
conditions $F_1\entailsClause c$ and ${\fn var}\ l \in \mathit{Vars}$ are used
in the \rulename{backjump} rule (the latter condition is important only for
termination). This weaker condition (inspired by the use of \sat{} engines in
\smt{} solvers) can be used also for the \rulename{unitPropagate} rule and leads
from the rule given in Definition \ref{def:dpll-system}, to its present version
(this change is not relevant for the system correctness).  The new version of
\rulename{unitPropagate} has much similarities with the \rulename{backjump} rule
--- the only difference is that the \rulename{backjump} rule always asserts the
implied literal to a proper prefix of the trail.

\medskip
\begin{exa}
\label{ex:backjump}
Let $F_0$ be the same formula as in Example \ref{ex:dpll}. One possible
$\rightarrowb$ trace is given below. Note that, unlike in the trace shown in
Example \ref{ex:dpll}, the decision literal $+4$ is removed from the trail
during backjumping, since it was detected to be irrelevant for the conflict,
resulting in a shorter trace. The deduction of backjump clauses (e.g., $[-2, -3,
-5]$) will be presented in Example \ref{ex:conflictanalysis}.

\noindent
\begin{tabular}{l|l}
rule                                                        & $M$                                                                                                                        \\ \hline
                                                            & $[\,]$                                                                                                                     \\ 
\rulename{decide} ($l = +1$),                               & $[\decision{+1}]$                                                                                                          \\
\rulename{unitPropagate} ($c = [-1, +2]$, $l = +2$)         & $[\decision{+1}, \nondecision{+2}]$                                                                                        \\
\rulename{unitPropagate} ($c = [-2, +3]$, $l = +3$)         & $[\decision{+1}, \nondecision{+2}, \nondecision{+3}]$                                                                      \\
\rulename{decide} ($l = +4$)                                & $[\decision{+1}, \nondecision{+2}, \nondecision{+3}, \decision{+4}]$                                                       \\
\rulename{decide} ($l = +5$)                                & $[\decision{+1}, \nondecision{+2}, \nondecision{+3}, \decision{+4}, \decision{+5}]$                                        \\
\rulename{unitPropagate} ($c = [-5, +6]$, $l = +6$)         & $[\decision{+1}, \nondecision{+2}, \nondecision{+3}, \decision{+4}, \decision{+5}, \nondecision{+6}]$                      \\
\rulename{unitPropagate} ($c = [-2, -5, +7]$, $l = +7$)     & $[\decision{+1}, \nondecision{+2}, \nondecision{+3}, \decision{+4}, \decision{+5}, \nondecision{+6}, \nondecision{+7}]$    \\
\rulename{backjump} ($c = [-2, -3, -5]$, $l = -5$)          & $[\decision{+1}, \nondecision{+2}, \nondecision{+3}, \nondecision{-5}]$                                                    \\
\rulename{unitPropagate} ($c = [-1, -3, +5, +7]$, $l = +7$) & $[\decision{+1}, \nondecision{+2}, \nondecision{+3}, \nondecision{-5}, \nondecision{+7}]$                                  \\
\rulename{backjump} ($c = [-1]$, $l = -1$)                  & $[\nondecision{-1}]$                                                                                                       \\
\rulename{decide} ($l = +2$)                                & $[\nondecision{-1}, \decision{+2}]$                                                                                        \\
\rulename{unitPropagate} ($c = [-2, +3]$, $l = +3$)         & $[\nondecision{-1}, \decision{+2}, \nondecision{+3}]$                                                                      \\
\end{tabular}

\begin{tabular}{l|l}
\rulename{decide} ($l = +4$)                                & $[\nondecision{-1}, \decision{+2}, \nondecision{+3}, \decision{+4}]$                                                       \\
\rulename{decide} ($l = +5$)                                & $[\nondecision{-1}, \decision{+2}, \nondecision{+3}, \decision{+4}, \decision{+5}]$                                        \\
\rulename{unitPropagate} ($c = [-5, +6]$, $l = +6$)         & $[\nondecision{-1}, \decision{+2}, \nondecision{+3}, \decision{+4}, \decision{+5}, \nondecision{+6}]$                      \\
\rulename{unitPropagate} ($c = [-2, -5, +7]$, $l = +7$)     & $[\nondecision{-1}, \decision{+2}, \nondecision{+3}, \decision{+4}, \decision{+5}, \nondecision{+6}, \nondecision{+7}]$    \\
\rulename{backjump} ($c = [-2, -3, -5])$                    & $[\nondecision{-1}, \decision{+2}, \nondecision{+3}, \nondecision{-5}]$                                                   \\
\rulename{decide} ($l = +4$)                                & $[\nondecision{-1}, \decision{+2}, \nondecision{+3}, \nondecision{-5}, \decision{+4}]$                                     \\
\rulename{decide} ($l = +6$)                                & $[\nondecision{-1}, \decision{+2}, \nondecision{+3}, \nondecision{-5}, \decision{+4}, \decision{+6}]$                      \\
\rulename{unitPropagate} ($c = [-3, -6, -7]$, $l = -7$)     & $[\nondecision{-1}, \decision{+2}, \nondecision{+3}, \nondecision{-5}, \decision{+4}, \decision{+6}, \nondecision{-7}]$    \\
\end{tabular}
\end{exa}
\bigskip

\subsection{Backjump Levels}

In Definition \ref{def:backjumping_system}, for the \rulename{backjump} rule to
be applicable, it is required that there is a level of the trail such that the
backjump clause is unit in the prefix to that level. The following definition
gives a stronger condition (used in modern \sat{} solvers) for a level ensuring
applicability of the \rulename{backjump} rule to that level.

\begin{defi}[Backjump level]
  A {\em backjump level} for the given backjump clause $c$ (false in $M$) is a
  level $level$ that is strictly less than the level of the last falsified
  literal from $c$, and greater or equal to the levels of the remaining literals
  from $c$:
  \begin{eqnarray*}
  {\fn isBackjumpLevel}\ level\ l\ c\ M\ & \iff & M\falsifies c \ \wedge\  \opp{l} = {\fn lastAssertedLiteral}\ \opp{c}\ M\ \wedge\ \\
  & & 0\leq level < {\fn level}\ \opp{l}\ M\ \wedge\ \\
  & & \forall\ l'.\ l'\in c\setminus l \implies {\fn level}\ \opp{l'}\ M  \leq level
  \end{eqnarray*}
\end{defi}

Using this definition, the \rulename{backjump} rule can be defined in a more 
concrete and more operational way.

\begin{small}
$$
\begin{array}{rll}
{\fn backjump}'\ (M_1, F_1)\ (M_2, F_2) &\iff& \exists c\ l\ level.\ F_1\entailsClause c \ \wedge\ {\fn var}\ l\in \mathit{Vars}\ \wedge\ \\
& & {\fn isBackjumpLevel}\ level\ l\ c\ M_1\ \wedge\ \\
& & M_2 = ({\fn prefixToLevel}\ level\ M_1)\,@\,\nondecision{l} \ \wedge\ F_2 = F_1
\end{array}
$$
\end{small}

Notice that, unlike in Definition \ref{def:backjumping_system}, it is required
that the backjump clause is false, so this new rule is applicable only in 
conflict situations.

It still remains unspecified how the clause $c$ is constructed. Also, it is
required to check whether the clause $c$ is false in the current trail $M$ and
implied by the current formula $F$. In Section \ref{sec:conflictAnalysis} it
will be shown that if a clause $c$ is built during a conflict analysis process,
these conditions will hold by construction and so it will not be necessary to
check them explicitly. Calculating the level of each literal from $c$ (required
for the backjump level condition) will also be avoided.

The following lemmas connect the \rulename{backjump} and \rulename{backjump}'
rules.

\begin{lem}
\label{lemma:BackjumpImplementationEnsuresUnitClause}
If:
\begin{enumerate}[\em(1)]
\packitems
\item ${\fn consistent}\ M$ (i.e., \inv{consistent} holds),
\item ${\fn unique}\ M$ (i.e., \inv{unique} holds),
\item ${\fn isBackjumpLevel}\ level\ l\ c\ M$,
\end{enumerate}
then ${\fn isUnit}\ c\ l\ ({\fn prefixToLevel}\ level\ M)$.
\end{lem}

\begin{lem}
\label{lemma:backjump_backjump'}
If a state $(M, F)$ satisfies the invariants and if ${\fn backjump}'\ (M, F)$
$(M', F')$, then ${\fn backjump}\ (M, F)\ (M', F')$.
\end{lem}

Because of the very close connection between the relations \rulename{backjump}
and \rulename{backjump}', we will not explicitly define two different transition
relations $\rightarrowb$. Most of the correctness arguments apply to both these
relations, and hence only differences will be emphasized.

Although there are typically many levels satisfying the backjump level
condition, (i.e., backjumping can be applied for each level between the level of
the last falsified literal from $c$ and the levels of the remaining literals
from $c$), usually it is applied to the lowest possible level, i.e., to the
level that is a backjump level such that there is no smaller level that is also
a backjump level.  The following definition introduces formally the notion of a
minimal backjump level.

\label{def:minimalBackjumpLevel}
\begin{defi}[${\fn isMinimalBackjumpLevel}$]
${\fn isMinimalBackjumpLevel}\ \mathit{level}\ l\ c\ M \iff $
  $${\fn isBackjumpLevel}\ \mathit{level}\ l\ c\ M\ \wedge
  (\forall\ \mathit{level}' < \mathit{level}.\ \neg {\fn
    isBackjumpLevel}\ \mathit{level}'\ l\ c\ M)$$
\end{defi}

Although most solvers use minimal levels when backjumping, this will be formally
required only for systems introduced in Section \ref{sec:restart}.

\subsection{Properties}
\label{subsec:backjumping_properties}

As in Section \ref{sec:DPLLSearch}, local properties of the transition 
rules in the form of certain invariants are used in proving properties of 
the transition system. 

\subsubsection{Invariants}
\label{subsec:backjump_invariants}

The invariants required for proving soundness, termination, and completeness of
the new system are the same as the invariants listed in Section
\ref{sec:DPLLSearch}. So, it is required to prove that the rules
\rulename{backjump} and the modified \rulename{unitPropagate} preserve all the
invariants. Therefore, Lemma \ref{lemma:invariantsHold} has to be updated to
address new rules and its proof has to be modified to reflect the changes in the
definition of the transition relation.

\subsubsection{Soundness and Termination}

The soundness theorem (Theorem \ref{thm:soundness}) has to be updated to address
the new rules, but its proof remains analogous to the one given in Section
\ref{sec:DPLLSearch}.

The termination theorem (Theorem \ref{thm:termination}) also has to be updated,
and its proof again remains analogous to the one given in \ref{sec:DPLLSearch}.
However, in addition to Lemma \ref{lemma:rightarrowsucc}, the following lemma
has to be used.

\begin{lem}
\label{lemma:backjumpSucc}
If ${\fn backjump}\ (M_1, F_1)\ (M_2, F_2)$, then $M_1 \succTr M_2$.
\end{lem}

This proof relies on the following property of the relation
$\succTr$.

\begin{lem}
\label{lemma:lexLessPrefixToLevel}
If $M$ is a trail and $P = {\fn prefixToLevel}\ \mathit{level}\ M$, such that
$0\leq \mathit{level} < {\fn currentLevel}\ M$, then $M\ \succTr\ P\ @\ \nondecision{l}$.
\end{lem}

\subsubsection{Completeness and Correctness}

Completeness of the system is proved partly in analogy with the completeness
proof of the system described in Section \ref{sec:DPLLSearch}, given in Theorem
\ref{thm:completeness}. When $(M, F)$ is a final state and $M \nfalsifies F$,
the proof remains the same as for Theorem \ref{thm:completeness}. When $(M, F)$
is a final state and $M \falsifies F$, for the new system it is not trivial that
this state is a rejecting state (i.e., it is not trivial that ${\fn
  decisions}\ M = [\,]$). Therefore, it has to be proved, given that the
invariants hold, that if backjumping is not applicable in a conflict situation
(when $M\falsifies F$), then ${\fn decisions}\ M = [\,]$ (i.e., if
${\fn decisions}\ M \neq [\,]$, then \rulename{backjump}' is applicable, and so
is \rulename{backjump}). The proof relies on the fact that a backjump clause
{\em may be} constructed only of all decision literals. This is the simplest way
to construct a backjump clause $c$ and in this case backjumping degenerates to
backtracking. The clause $c$ constructed in this way meets sufficient (but, of
course, not necessary) conditions for the applicability of \rulename{backjump'}
(and, consequently, by Lemma \ref{lemma:backjump_backjump'}, for the
applicability of \rulename{backjump}).

\begin{lem}
\label{lemma:completeness_reject}
If for a state $(M,F)$ it holds that:
\begin{enumerate}[\em(1)]
\packitems
\item ${\fn consistent}\ M$ (i.e., \inv{consistent} holds),
\item ${\fn unique}\ M$ (i.e., \inv{unique} holds),
\item $\forall l.\ l \in M \implies F\ @\ ({\fn decisionsTo}\ l\ M)
  \entailsLiteral l$ (i.e., \inv{impliedLits} holds),
\item ${\fn vars}\ M \subseteq \mathit{Vars}$ (i.e., \inv{varsM} holds),
\item $M\falsifies F$,
\item ${\fn decisions}\ M \neq [\,]$,
\end{enumerate}
then there is a state $(M',F')$ such that ${\fn backjump}'\ (M, F)\ (M',
F')$.
\end{lem}

To ensure applicability of Lemma \ref{lemma:completeness_reject}, the new
version of the completeness theorem (Theorem \ref{thm:completeness}) requires
that the invariants hold in the current state.  Since, by Lemma
\ref{lemma:backjump_backjump'}, ${\fn backjump}'\ (M, F)\ (M', F')$ implies
${\fn backjump}\ (M, F)$ $(M', F')$, the following completeness theorem holds
for both transition systems presented in this section (using the rule
\rulename{backjump'} or the rule \rulename{backjump}).

\begin{thm}[Completeness for $\rightarrowb$]
If $([\,], F_0) \rightarrowb^* (M, F)$, and $(M, F)$ is a final
state, then $(M, F)$ is either accepting or rejecting.
\end{thm}

\begin{proof}
Let $(M, F)$ be a final state. By Lemma \ref{lemma:invariantsHold}, all
invariants hold in $(M, F)$. Also, it holds that either $M\falsifies F$ or
$M\nfalsifies F$.

If $M\nfalsifies F$, since \rulename{decide} is not applicable, $(M,F)$ is an
accepting state.

If $M\falsifies F$, assume that ${\fn decisions}\ M \neq [\,]$. By Lemma
\ref{lemma:completeness_reject}, there is a state $(M',F')$ such that ${\fn
  backjump}'\ (M, F)\ (M', F')$. This contradicts the assumption that $(M,F)$ is
a final state.  Therefore, ${\fn decisions}\ M = [\,]$, and since $M\falsifies
F$, $(M, F)$ is a rejecting state.
\end{proof}

Correctness of the system is a consequence of soundness, termination,
and completeness, in analogy with Theorem \ref{thm:correctnes}.

\section{Learning and Forgetting}
\label{sec:learningForgetting}

In this section we briefly describe a system obtained from the system introduced
in Section \ref{sec:backjumping} by adding two new transition rules. These rules
will have a significant role in more complex systems discussed in the following
sections.

\subsection{States and Rules}
\label{subsec:learn_rules}

The relation $\rightarrowb$ introduced in Section \ref{sec:backjumping}
is extended by the two following transition rules (introduced in the form 
of relations over states).

\begin{small}
\begin{defi}[Transition rules] \hfill
\begin{eqnarray*}
{\fn learn}\ (M_1, F_1)\ (M_2, F_2) & \iff & \exists\ c\ .\quad F_1\entailsClause c  \ \wedge\ {\fn vars}\ c\subseteq \mathit{Vars}\ \wedge\ \\
& & \quad \qquad F_2 = F_1\ @\ c \ \wedge\ M_2 = M_1
\end{eqnarray*}
\vspace{-0.25cm}
\begin{eqnarray*}
{\fn forget}\ (M_1, F_1)\ (M_2, F_2) & \iff & \exists\ c\ .\ F_1\setminus c \entailsClause c  \ \wedge\ \\
& & \qquad F_2 = F_1\setminus c \ \wedge\ M_2 = M_1
\end{eqnarray*}
The extended transition system will be denoted by $\rightarrowl$.
\end{defi}
\end{small}

The \rulename{learn} rule is defined very generally. It is not specified how to
construct the clause $c$ --- typically, only clauses resulting from the conflict
analysis process (Section \ref{sec:conflictAnalysis}) are learnt. This is the
only rule so far that changes $F$, but the condition $F\entailsClause c$ ensures
that it always remains logically equivalent to the initial formula $F_0$. The
condition ${\fn vars}\ c\subseteq \mathit{Vars}$ is relevant only for ensuring
termination.

The \rulename{forget} rule changes the formula by removing a clause that is
implied by all other clauses (i.e., is redundant). It is also not specified
how this clause $c$ is chosen.

\medskip
\begin{exa}
\label{ex:learn}
Let $F_0$ be a formula from Example \ref{ex:dpll}. A possible $\rightarrowl$
trace is given by (note that, unlike in the trace shown in Example
\ref{ex:backjump}, a clause $[-1, -2, -3]$ is learnt and used afterwards for
unit propagation in another part of the search tree, eventually leading to a
shorter trace):

\noindent
\resizebox{\linewidth}{!}{
\begin{tabular}{l||l|l}
rule                                                        & $M$                                                                                                                        & $F$\\ \hline
                                                            & $[\,]$                                                                                                                     & $F_0$\\ 
\rulename{decide} ($l = +1$),                               & $[\decision{+1}]$                                                                                                          & $F_0$\\
\rulename{unitPropagate} ($c = [-1, +2]$, $l = +2$)         & $[\decision{+1}, \nondecision{+2}]$                                                                                        & $F_0$\\
\rulename{unitPropagate} ($c = [-2, +3]$, $l = +3$)         & $[\decision{+1}, \nondecision{+2}, \nondecision{+3}]$                                                                      & $F_0$\\
\rulename{decide} ($l = +4$)                                & $[\decision{+1}, \nondecision{+2}, \nondecision{+3}, \decision{+4}]$                                                       & $F_0$\\
\rulename{decide} ($l = +5$)                                & $[\decision{+1}, \nondecision{+2}, \nondecision{+3}, \decision{+4}, \decision{+5}]$                                        & $F_0$\\
\rulename{unitPropagate} ($c = [-5, +6]$, $l = +6$)         & $[\decision{+1}, \nondecision{+2}, \nondecision{+3}, \decision{+4}, \decision{+5}, \nondecision{+6}]$                      & $F_0$\\
\rulename{unitPropagate} ($c = [-2, -5, +7]$, $l = +7$)     & $[\decision{+1}, \nondecision{+2}, \nondecision{+3}, \decision{+4}, \decision{+5}, \nondecision{+6}, \nondecision{+7}]$    & $F_0$\\
\rulename{backjump} ($c = [-2, -3, -5]$, $l = -5$)          & $[\decision{+1}, \nondecision{+2}, \nondecision{+3}, \nondecision{-5}]$                                                    & $F_0$\\
\rulename{learn} ($c = [-2, -3, -5]$)                       & $[\decision{+1}, \nondecision{+2}, \nondecision{+3}, \nondecision{-5}]$                                                    & $F_0\,@\,[-2, -3, -5]$\\
\rulename{unitPropagate} ($c = [-1, -3, +5, +7]$, $l = +7$) & $[\decision{+1}, \nondecision{+2}, \nondecision{+3}, \nondecision{-5}, \nondecision{+7}]$                                  & $F_0\,@\,[-2, -3, -5]$\\
\rulename{backjump} ($c = [-1]$, $l = -1$)                  & $[\nondecision{-1}]$                                                                                                       & $F_0\,@\,[-2, -3, -5]$\\
\rulename{decide} ($l = +2$)                                & $[\nondecision{-1}, \decision{+2}]$                                                                                        & $F_0\,@\,[-2, -3, -5]$\\
\rulename{unitPropagate} ($c = [-2, +3]$, $l = +3$)         & $[\nondecision{-1}, \decision{+2}, \nondecision{+3}]$                                                                      & $F_0\,@\,[-2, -3, -5]$\\
\rulename{unitPropagate} ($c = [-2, -3, -5]$, $l = -5$)     & $[\nondecision{-1}, \decision{+2}, \nondecision{+3}, \nondecision{-5}]$                                                    & $F_0\,@\,[-2, -3, -5]$\\
\rulename{decide} ($l = +4$)                                & $[\nondecision{-1}, \decision{+2}, \nondecision{+3}, \nondecision{-5}, \decision{+4}]$                                     & $F_0\,@\,[-2, -3, -5]$\\
\rulename{decide} ($l = +6$)                                & $[\nondecision{-1}, \decision{+2}, \nondecision{+3}, \nondecision{-5}, \decision{+4}, \decision{+6}]$                      & $F_0\,@\,[-2, -3, -5]$\\
\rulename{unitPropagate} ($c = [-3, -6, -7]$, $l = -7$)     & $[\nondecision{-1}, \decision{+2}, \nondecision{+3}, \nondecision{-5}, \decision{+4}, \decision{+6}, \nondecision{-7}]$    & $F_0\,@\,[-2, -3, -5]$\\
\end{tabular}
}
\end{exa}

\subsection{Properties}

The new set of rules preserves all the invariants given in Section
\ref{subsec:DPLLInvariants}. Indeed, since \rulename{learn} and
\rulename{forget} do not change the trail $M$, all invariants about the trail
itself are trivially preserved by these rules. It can be proved that
\inv{equiv}, \inv{varsF} and \inv{impliedLiterals} also hold for the new rules.

Since the invariants are preserved in the new system, soundness is proved as in
Theorem \ref{thm:soundness}. Completeness trivially holds, since introducing new
rules to a complete system cannot compromise its completeness. However, the
extended system is not terminating since the \rulename{learn} and
\rulename{forget} rules can by cyclically applied. Termination could be ensured
with some additional restrictions. Specific learning, forgetting and backjumping
strategies that ensure termination will be defined and discussed in Sections
\ref{sec:conflictAnalysis} and \ref{sec:restart}.


\section{Conflict Analysis}
\label{sec:conflictAnalysis}

The backjumping rules, as defined in Section \ref{sec:backjumping}, are very
general. If backjump clauses faithfully reflect the current conflict, they
typically lead to significant pruning of the search space. In this section we
will consider a transition system that employs conflict analysis in order to
construct backjump clauses, which can be (in addition) immediately learned (by
the rule \rulename{learn}).

\subsection{States and Rules}

The system with conflict analysis requires extending the definition of state
introduced in Section \ref{sec:DPLLSearch}.

\begin{defi}[State]
A state of the system is a four-tuple $(M, F, C, \mathit{cflct})$, where $M$ is
a trail, $F$ is a formula, $C$ is a clause, and $\mathit{cflct}$ is a Boolean
variable. A state $([\,], F_0, [\,], \bot)$ is an {\em initial state} for the
input formula $F_0$.
\end{defi}

Two new transition rules \rulename{conflict} and \rulename{explain} are defined
in the form of relations over states. In addition, the existing rules are 
updated to map four-tuple states to four-tuple states.
\bigskip

\begin{small}
\begin{defi}[Transition rules]\hfill

\noindent${\fn decide}\ (M_1, F_1, C_1, \mathit{cflct}_1)\ (M_2, F_2, C_2,
\mathit{cflct}_2) \iff$
\vspace{-1mm}
\begin{eqnarray*}
\exists l.& &{\fn var}\ l\in \mathit{DecVars}\ \wedge\ l\notin M_1\ \wedge\ \opp{l}\notin M_1\ \wedge\ \\
& &M_2 = M_1\,@\,\decision{l} \ \wedge\ F_2 = F_1\ \wedge\ C_2 = C_1\ \wedge\ \mathit{cflct}_2 = \mathit{cflct}_1
\end{eqnarray*}

\noindent${\fn unitPropagate}\ (M_1, F_1, C_1, \mathit{cflct}_1)\ (M_2, F_2, C_2, \mathit{cflct}_2) \iff$
\vspace{-1mm}
\begin{eqnarray*}
\exists c\ l.& &F_1\entailsClause c\ \wedge\ {\fn var}\ l\in {\fn
  vars}\ \mathit{Vars}\ \wedge\ {\fn
  isUnit}\ c\ l\ M_1\ \wedge\ \\ & &M_2 = M_1\,@\,\nondecision{l}
\ \wedge\ F_2 = F_1\ \wedge\ C_2 = C_1 \ \wedge\ \mathit{cflct}_2 =
\mathit{cflct}_1
\end{eqnarray*}

\noindent${\fn conflict}\ (M_1, F_1, C_1, \mathit{cflct}_1)\ (M_2, F_2, C_2, \mathit{cflct}_2) \iff$
\vspace{-1mm}
\begin{eqnarray*}
\exists c.& &\mathit{cflct}_1 = \bot\ \wedge\ F_1\entailsClause c\ \wedge\ M_1\falsifies c \wedge\ \\
& &M_2 = M_1 \ \wedge\ F_2 = F_1\ \wedge\ C_2 = c \ \wedge\ \mathit{cflct}_2 = \top
\end{eqnarray*}

\noindent${\fn explain}\ (M_1, F_1, C_1, \mathit{cflct}_1)\ (M_2, F_2, C_2, \mathit{cflct}_2) \iff$
\vspace{-1mm}
\begin{eqnarray*}
\exists\ l\ c.& &\mathit{cflct}_1 = \top\ \wedge\ l\in C_1\ \wedge\ {\fn isReason}\
c\ \opp{l}\ M_1\ \wedge\ F_1\entailsClause c \ \wedge\\
& &M_2 = M_1 \ \wedge\ F_2 = F_1\ \wedge\ C_2 = {\fn resolve}\ C_1\ c\ l \ \wedge\ \mathit{cflct}_2 = \top
\end{eqnarray*}

\noindent${\fn backjump}\ (M_1, F_1, C_1, \mathit{cflct}_1)\ (M_2, F_2, C_2, \mathit{cflct}_2) \iff$
\vspace{-1mm}
\begin{eqnarray*}
\exists l\ level.& &\mathit{cflct}_1 = \top\ \wedge\ {\fn isBackjumpLevel}\ level\ l\ C_1\ M_1\ \wedge\ \\
& &M_2 = ({\fn prefixToLevel}\ level\ M_1)\,@\,\nondecision{l} \ \wedge\ F_2 = F_1\
\wedge\ \\
& &  C_2 = [\,] \ \wedge\ \mathit{cflct}_2 = \bot
\end{eqnarray*}

\noindent${\fn learn}\ (M_1, F_1, C_1, \mathit{cflct}_1)\ (M_2, F_2, C_2, \mathit{cflct}_2) \iff$
\vspace{-1mm}
\begin{eqnarray*}
& & \mathit{cflct}_1 =\top  \ \wedge\ C_1\notin F_1\\
& & M_2 = M_1\ \wedge\ F_2 = F_1\ @\ C_1 \ \wedge\ C_2 = C_1 \ \wedge\ \mathit{cflct}_2 = \mathit{cflct}_1
\end{eqnarray*}
\end{defi}
\end{small}

The relation $\rightarrowc$ is defined as in Definition \ref{def:rightarrow},
but using the above list of rules. The definition of outcome states also has to
be updated.

\begin{defi}[Outcome states]
A state is an {\em accepting state} if $\mathit{cflct} = \bot$,
$M\nfalsifies F$ and there is no literal such that ${\fn var}\ l \in
\mathit{DecVars}$, $l\notin M$ and $\opp{l}\notin M$.

A state is a {\em rejecting state} if $\mathit{cflct} = \top$ and
$C = [\,]$.
\end{defi}

\medskip
\begin{exa}
\label{ex:conflictanalysis}
Let $F_0$ be a formula from Example \ref{ex:dpll}. A possible $\rightarrowc$
trace (shown up to the first application of \rulename{backjump}) is given (due
to the lack of space, the $F$ component of the state is not shown).

\noindent
\resizebox{\linewidth}{!}{
\begin{tabular}{l||l|l|l}
rule                                                        & $M$                                                                                                                        & $cflct$ & $C$ \\ \hline
                                                            & $[\,]$                                                                                                                     & $\bot$  & $[]$\\ 
\rulename{decide} ($l = +1$),                               & $[\decision{+1}]$                                                                                                          & $\bot$  & $[]$\\
\rulename{unitPropagate} ($c = [-1, +2]$, $l = +2$)         & $[\decision{+1}, \nondecision{+2}]$                                                                                        & $\bot$  & $[]$\\
\rulename{unitPropagate} ($c = [-2, +3]$, $l = +3$)         & $[\decision{+1}, \nondecision{+2}, \nondecision{+3}]$                                                                      & $\bot$  & $[]$\\
\rulename{decide} ($l = +4$)                                & $[\decision{+1}, \nondecision{+2}, \nondecision{+3}, \decision{+4}]$                                                       & $\bot$  & $[]$\\
\rulename{decide} ($l = +5$)                                & $[\decision{+1}, \nondecision{+2}, \nondecision{+3}, \decision{+4}, \decision{+5}]$                                        & $\bot$  & $[]$\\
\rulename{unitPropagate} ($c = [-5, +6]$, $l = +6$)         & $[\decision{+1}, \nondecision{+2}, \nondecision{+3}, \decision{+4}, \decision{+5}, \nondecision{+6}]$                      & $\bot$  & $[]$\\
\rulename{unitPropagate} ($c = [-2, -5, +7]$, $l = +7$)     & $[\decision{+1}, \nondecision{+2}, \nondecision{+3}, \decision{+4}, \decision{+5}, \nondecision{+6}, \nondecision{+7}]$    & $\bot$  & $[]$\\
\rulename{conflict} ($c = [-3, -6, -7]$)                    & $[\decision{+1}, \nondecision{+2}, \nondecision{+3}, \decision{+4}, \decision{+5}, \nondecision{+6}, \nondecision{+7}]$    & $\top$  & $[-3, -6, -7]$\\
\rulename{explain} ($l = -7$, $c = [-2, -5, +7]$)           & $[\decision{+1}, \nondecision{+2}, \nondecision{+3}, \decision{+4}, \decision{+5}, \nondecision{+6}, \nondecision{+7}]$    & $\top$  & $[-2, -3, -5, -6]$\\
\rulename{explain} ($l = -6$, $c = [-5, +6]$)               & $[\decision{+1}, \nondecision{+2}, \nondecision{+3}, \decision{+4}, \decision{+5}, \nondecision{+6}, \nondecision{+7}]$    & $\top$  & $[-2, -3, -5]$\\
\rulename{learn} ($c = [-2, -3, -5]$)                       & $[\decision{+1}, \nondecision{+2}, \nondecision{+3}, \decision{+4}, \decision{+5}, \nondecision{+6}, \nondecision{+7}]$    & $\top$  & $[-2, -3, -5]$\\
\rulename{backjump} ($c = [-2, -3, -5]$, $l = -5$)          & $[\decision{+1}, \nondecision{+2}, \nondecision{+3}, \nondecision{-5}]$                                                    & $\bot$  & $[]$\\
\end{tabular}
}
\end{exa}

\subsection{Unique Implication Points (UIP)}

\sat{} solvers employ different strategies for conflict analysis. The most
widely used is a \emph{1-UIP} strategy, relying on a concept of {\em unique
  implication points (UIP)} (often expressed in terms of implication graphs
\cite{grasp}). Informally, a clause $c$, false in the trail $M$, satisfies the
UIP condition if there is exactly one literal in $c$ that is on the highest
decision level of $M$. The UIP condition is very easy to check. The \emph{1-UIP}
strategy requires that the rule \rulename{explain} is always applied to the last
literal false in $M$ among literals from $c$, and that backjumping is applied as
soon as $c$ satisfies the UIP condition.

\begin{defi}[Unique implication point]
A clause $c$ that is false in $M$ has a {\em unique implication point}, denoted
by ${\fn isUIP}\ l\ c\ M$, if the level of the last literal $l$ from $c$ that is
false in $M$ is strictly greater than the level of the remaining literals from
$c$ that are false in $M$:

  \begin{eqnarray*}
{\fn isUIP}\ l\ c\ M\  & \iff & M\falsifies c \ \wedge\ \opp{l} = {\fn lastAssertedLiteral}\ \opp{c}\ M\ \wedge\ \\
  & & \forall\ l'.\ l'\in c\setminus l \implies {\fn level}\ \opp{l'}\ M  <  {\fn level}\ \opp{l}\ M
  \end{eqnarray*}

\end{defi}

The following lemma shows that, if there are decision literals in $M$, if a
clause has a unique implication point, then there is a corresponding backjump
level, and consequently, the \rulename{backjump} rule is applicable.

\begin{lem}
\label{lemma:isUIPExistsBackjumpLevel}
If ${\fn unique}\ M$ (i.e., \inv{unique} holds), then
$${\fn isUIP}\ l\ c\ M  \; \wedge \; {\fn level}\ \opp{l}\ M > 0 \iff
\exists \ level. \ {\fn isBackjumpLevel}\ level\ l\ c\ M
$$
\end{lem}

Therefore, the guard ${\fn isBackjumpLevel}\ \mathit{level}\ l\ c\ M$ in the
definition of the \rulename{backjump} rule can be replaced by the stronger
conditions ${\fn isUIP}\ l\ c\ M$ and ${\fn level}\ \opp{l}\ M > 0$.  In that
case, the backjump level $\mathit{level}$ has to be explicitly calculated (as in
the proof of the previous lemma).

The UIP condition is trivially satisfied when the clause $c$ consists only of
opposites of decision literals from the trail (a similar construction of $c$ was
already used in the proof of Lemma \ref{lemma:completeness_reject}).

\begin{lem}
\label{lemma:allDecisionsThenUIP}
If it holds that:
\begin{enumerate}[\em(1)]
\packitems
\item ${\fn unique}\ M$ (i.e., \inv{unique} holds),
\item $\opp{c}\ \subseteq\ {\fn decisions}\ M$,
\item $\opp{l} = {\fn lastAssertedLiteral}\ \opp{c}\ M$,
\end{enumerate}
then ${\fn isUIP}\ l\ c\ M$.
\end{lem}

\subsection{Properties}

Properties of the new transition system will be again proved using invariants
introduced in Section \ref{sec:DPLLSearch}, but they have to be updated to
reflect the new definition of states. In addition, three new invariants will be
used.

\subsubsection{Invariants}

In addition to the invariants from Section \ref{sec:DPLLSearch}, three
new invariants are used.

\begin{small}
\noindent
\begin{tabular}{l@{ }l}
\defineInvariant{inv:Cfalse}
\inv{Cfalse}: & $\mathit{cflct} \implies M\falsifies C$\\
\defineInvariant{inv:Centailed}
\inv{Centailed}: & $\mathit{cflct} \implies F\entailsClause C$ \\
\defineInvariant{inv:reasonClauses}
{\inv{reasonClauses}}: & $\forall\ l.\ l\in M\ \wedge\ l\notin {\fn decisions}\ M \implies \exists\ c.\ {\fn isReason}\ c\ l\ M\ \wedge\ F\entailsClause c$
\end{tabular}
\end{small}

The first two invariants ensure that during the conflict analysis process, the
conflict analysis clause $C$ is a consequence of $F$ and that $C$ is false in
$M$. The third invariant ensures existence of clauses that are reasons of
literal propagation (these clauses enable application of the \rulename{explain}
rule). By the rules \rulename{unitPropagate} and \rulename{backjump} literals
are added to $M$ only as implied literals and in both cases propagation is
performed using a clause that is a reason for propagation, so this clause can be
associated to the implied literal, and afterwards used as its reason.

Lemma \ref{lemma:invariantsHold} again has to be updated to address new rules
and its proof has to be modified to reflect the changes in the definition of the
relation $\rightarrowc$.

\subsubsection{Soundness}

Although the soundness proof for unsatisfiable formulae could be again based on
Lemma \ref{lemma:InvariantImpliedLiteralsAndFormulaFalse}, this time it will be
proved in an alternative, simpler way (that does not rely on the invariant 
\inv{impliedLits}), that was not possible in previous sections.

\begin{lem}
\label{lemma:unsatReportCEmpty}
If there is a rejecting state $(M,F,C,\mathit{cflct})$ such that it holds
\begin{enumerate}[\em(1)]
\packitems
\item $F \equiv F_0$, (i.e., \inv{equiv} holds)
\item $\mathit{cflct} \implies F\entailsClause C$ (i.e., \inv{Centailed}
  holds),
\end{enumerate}
then $F_0$ is unsatisfiable (i.e., $\neg({\fn sat}\ F_0)$).
\end{lem}

\begin{thm}[Soundness for $\rightarrowc$]
If $([\,], F_0, [\,], \bot)\rightarrowc^* (M, F, C,
\mathit{cflct})$, then:
\begin{enumerate}[\em(1)]
\packitems
\item If $\mathit{DecVars} \supseteq {\fn vars}\ F_0$ and $(M,F)$
  is an accepting state, then the formula is $F_0$ satisfiable and $M$ is its
  model (i.e., ${\fn sat}\ F_0$ and ${\fn model}\ M\ F_0$).
\item If $(M,F, C, \mathit{cflct})$ is a rejecting state, then
  the formula $F_0$ is unsatisfiable (i.e., $\neg({\fn sat}\ F_0)$).
\end{enumerate}
\end{thm}

\begin{proof}
By Lemma \ref{lemma:invariantsHold}, all the invariants hold in the state $(M,
F, C, \mathit{cflct})$.
\begin{enumerate}[(1)]
\packitems
\item All conditions of Lemma \ref{lemma:satReport} are met (adapted to the new
  defintion of state), so ${\fn sat}\ F_0$ and ${\fn model}\ M\ F_0$.
\item All conditions of Lemma \ref{lemma:unsatReportCEmpty} are met, so $\neg
  ({\fn sat}\ F_0)$.
\end{enumerate}
\vspace{-0.6cm}
\end{proof}

\subsubsection{Termination}
\label{sec:analizaKonflikataZaustavljanje}

Termination of the system with conflict analysis will be proved by using a
suitable well-founded ordering that is compatible with the relation
$\rightarrowc$, i.e., an ordering $\succ$ such that $s\rightarrowc s'$ yields
$s\succ s'$, for any two states $s$ and $s'$.  This ordering will be constructed
as a lexicographic combination of four simpler orderings, one for each state
component.

The rules \rulename{decide}, \rulename{unitPropagate}, and \rulename{backjump}
change $M$ and no other state components. If a state $s$ is in one of these
relations with the state $s'$ then $M \succTrrestrict[\mathit{Vars}]
M'$ (for the ordering $\succTrrestrict[\mathit{Vars}]$, introduced in Section
\ref{sec:DPLLtermination}). 

The ordering $\succTrrestrict[\mathit{Vars}]$ cannot be used alone for proving
termination of the system, since the rules \rulename{conflict},
\rulename{explain}, and \rulename{learn} do not change $M$ (and, hence, if a
state $s$ is transformed into a state $s'$ by one of these rules, then it does
not hold that $M \succTrrestrict[Vars] M'$). For each of these rules, a specific
well-founded ordering will be constructed and it will be proved that these rules
decrease state components with respect to those orderings.

The ordering $\succBool$ will be used for handling the state component
$\mathit{cflct}$ and the rule \rulename{conflict} (the rule \rulename{explain}
changes the state component $\mathit{cflct}$, but also the state component $C$,
so it will be handled by another ordering).  Given properties of the ordering
$\succBool$ are proved trivially.

\begin{defi}[$\succBool$]
$b_1 \succBool b_2 \iff b_1 = \bot\ \wedge\ b_2=\top.$
\end{defi}

\begin{lem}
\label{lemma:conflictSucc}
\begin{small}
If $ {\fn conflict}\ (M_1, F_1, C_1, \mathit{cflct}_1)\ (M_2, F_2, C_2,
  \mathit{cflct}_2)$, then $\mathit{cflct}_1 \succBool \mathit{cflct}_2$.
\end{small}
\end{lem}

\begin{lem}
\label{lemma:wfConflictSucc}
The ordering $\succBool$ is well-founded.
\end{lem}

An ordering over clauses (that are the third component of the states) should be
constructed such that the rule \rulename{explain} decreases the state component
$C$ with respect to that ordering. Informally, after each application of the
rule \rulename{explain}, a literal $l$ of the clause $C$ that is (by
\inv{Cfalse}) false in $M$ is replaced by several other literals that are again
false in $M$, but for them it holds that their opposite literals precede the
literal $\opp{l}$ in $M$ (since reason clauses are used). Therefore, the
ordering of literals in the trail $M$ defines an ordering of clauses false in
$M$. The ordering over clauses will be a multiset extension of the relation
$\precedes{M}$ induced by the ordering of literals in $M$ (Definition
\ref{def:precedesList}). Each explanation step removes a literal from $C$ and
replaces it with several literals that precede it in $M$. To avoid multiple
occurrences of a literal in $C$, duplicates are removed. Solvers usually perform
this operation explicitly and maintain the condition that $C$ does not contain
duplicates. However, our ordering does not require this restriction and
termination is ensured even without it.

\begin{defi}[$\succC{M}$]
For a trail $M$, 
$C_1 \succC{M} C_2 \iff \multisetof{{\fn remDups}\ \opp{C_2}}
\mult{\precedes{M}} \multisetof{{\fn remDups}\ \opp{C_1}}$.
\end{defi}

\begin{lem}
\label{lemma:wfClauseSucc}
For any trail $M$, the ordering $\succC{M}$ is well-founded.
\end{lem}

The following lemma ensures that each explanation step decreases the conflict
clause in the ordering $\succC{M}$, for the current trail $M$. This ensures that
each application of the \rulename{explain} rule decreases the state with respect
to this ordering.

\begin{lem}
\label{lemma:succResolve}
If $l\in C$ and ${\fn isReason}\ c\ \opp{l}\ M$, then $C\succC{M}
{\fn resolve}\ C\ c\ l$.
\end{lem}

\begin{lem}
\label{lemma:explainSucc}
If ${\fn explain}\ (M, F, C_1, \mathit{cflct})\ (M, F, C_2, \mathit{cflct})$,
  then $C_1 \succC{M} C_2$.
\end{lem}

The rule \rulename{learn} changes the state component $F$ (i.e., it adds a
clause to the formula) and it requires constructing an ordering over formulae.

\begin{defi}[$\succF{C}$]
For any clause $C$, $F_1 \succF{C} F_2 \iff  C\notin F_1 \wedge C\in F_2.$
\end{defi}

\begin{lem}
\label{lemma:wfFormulaSucc}
For any clause $C$, the ordering $\succF{C}$ is well-founded.
\end{lem}

By the definition of the \rulename{learn} rule, it holds that $C\notin F_1$ and
$C\in F_2$, so the following lemma trivially holds.

\begin{lem}
\label{lemma:learnSucc}
If ${\fn learn}\ (M, F_1, C, \mathit{cflct})\ (M, F_2, C, \mathit{cflct})$,
then $F_1 \succF{C} F_2$.
\end{lem}

\begin{thm}[Termination for $\rightarrowc$]
  If the set $\mathit{DecVars}$ is finite, for any formula $F_0$, the relation
  $\rightarrowc$ is well-founded on the set of states $s$ such that
  $s_0\rightarrowc^* s$, where $s_0$ is the initial state for $F_0$.
\end{thm}

\begin{proof}
  Let $\succ$ be a (parametrized) lexicographic product (Definition
  \ref{def:lexProd}), i.e., let
  
  $
  \succ\ \ \equiv\ \ \lexprodp{\lexprodp{\lexprod{\succTrrestrict[\mathit{Vars}]}{\succBool}}{\left(\lambda
      s. \succC{M_s}\right)}}{\left(\lambda s. \succF{C_s}\right)},$ 
  
  \noindent    
  where
  $M_s$ is the trail in the state $s$, and $C_s$ is the conflict clause in the
  state $s$.
  By Proposition \ref{prop:wfLeksikografskaKombinacijaUredjenja} and Lemmas
  \ref{lemma:wfTrailSuccRestricted}, \ref{lemma:wfConflictSucc},
  \ref{lemma:wfClauseSucc}, and \ref{lemma:wfFormulaSucc}, the relation $\succ$
  is well-founded.
  If the invariants hold in the state $(M_1, F_1, C_1, \mathit{cflct}_1)$ and if
  $(M_1, F_1, C_1, \mathit{cflct}_1) \rightarrowc (M_2, F_2, C_2,
  \mathit{cflct}_2)$, then $(M_1, \mathit{cflct}_1, C_1, F_1) \succ (M_2,
  \mathit{cflct}_2, C_2, F_2)$. Indeed, by Lemma \ref{lemma:rightarrowsucc}, the
  rules \rulename{decide}, \rulename{unitPropagate} and \rulename{backjump}
  decrease $M$ in the ordering, the rule \rulename{conflict} does not change $M$
  but (by Lemma \ref{lemma:conflictSucc}) decreases $\mathit{cflct}$, the rule
  \rulename{explain} does not change $M$ nor $\mathit{cflct}$, but (by Lemma
  \ref{lemma:explainSucc}) decreases $C$, and the rule \rulename{learn} does not
  change $M$, $\mathit{cflct}$, nor $C$, but (by Lemma \ref{lemma:learnSucc})
  decreases $F$.

  Then the theorem holds by Proposition \ref{prop:wfInvImage} (where ${\fn f}$
  is a permutation mapping $(M,$ $F,$ $C,$ $\mathit{cflct})$ to $(M, \mathit{cflct},
  C, F)$).
\end{proof}

\subsubsection{Completeness and Correctness}

Completeness requires that all final states are outcome states, and the
following two lemmas are used to prove this property.

\begin{lem}
\label{lemma:finalConflictingState}
If for the state $(M, F, C, \mathit{cflct})$ it holds that:
\begin{enumerate}[\em(1)]
\packitems
\item $\mathit{cflct} = \top$,
\item ${\fn unique}\ M$ (i.e., \inv{unique} holds),
\item $\mathit{cflct} \implies M\falsifies C$ (i.e., \inv{Cfalse} holds),
\item the rules \rulename{explain} and \rulename{backjump} are not applicable,
\end{enumerate}
then the state $(M, F, C, \mathit{cflct})$ is a rejecting state and $C = [\,]$.
\end{lem}

\begin{lem}
\label{lemma:finalNonConflictingState}
If in the state $(M, F, C, \mathit{cflct})$ it holds that $\mathit{cflct} =
\bot$ and the rule \rulename{conflict} is not applicable, then the
state $(M, F, C, \mathit{cflct})$ is an accepting state and $M\nfalsifies F$.
\end{lem}

\begin{thm}[Completeness for $\rightarrowc$]
For any formula $F_0$, if $([\,], F_0, [\,], \bot) \rightarrowc^* (M, F, C,
\mathit{cflct})$, and if the state $(M, F, C, \mathit{cflct})$ is final, then it
is either accepting or rejecting.
\end{thm}
\begin{proof}
Since the state $(M, F, C, \mathit{cflct})$ is reachable from the initial state,
by Lemma \ref{lemma:invariantsHold}, all the invariants hold in this state,
including ${\fn unique}\ M$ (i.e., \inv{unique}), and $\mathit{cflct} \implies
M\falsifies C$ (i.e., \inv{Cfalse}). In the state $(M, F, C, \mathit{cflct})$,
it holds that either $\mathit{cflct} = \top$ or $\mathit{cflct} = \bot$. If 
$\mathit{cflct} = \bot$, since the rule \rulename{decide} is not applicable 
(as the state is final), by Lemma \ref{lemma:finalNonConflictingState}, the 
state $(M, F, C, \mathit{cflct})$ is a rejecting state. If 
$\mathit{cflct} = \top$, since the rule \rulename{conflict} is not applicable 
(as the state is final) by Lemma \ref{lemma:finalConflictingState}, the state 
is an accepting state.
\end{proof}

Correctness of the system is proved in analogy with Theorem
\ref{thm:correctnes}.

\section{Restarting and Forgetting}
\label{sec:restart}

In this section we extend the previous system with restarting and
forgetting. The most challenging task with restarting is to ensure termination.

Many solvers use restarting and forgetting schemes that apply restarting with
\emph{increasing periodicity} and there are theoretical results ensuring total
correctness of these \cite{Krstic-Frocos07,NieOT-Jacm06}. However, modern
solvers also use \emph{aggressive restarting schemes} (e.g., Luby restarts) that
apply the \rulename{restart} rule very frequently, but there are no
corresponding theoretical results that ensure termination of these schemes. In
this section we will formulate a system that allows application of the restart
rule after each conflict and show that this (weakly constrained, hence 
potentially extremely frequent) scheme also ensures termination.

\subsection{States and Rules}

Unlike previous systems that tend to be as abstract as possible, this system
aims to precisely describe the behaviour of modern \sat{} solvers. For example,
only learnt clauses can be forgotten. So, to aid the \rulename{forget} rule, the
formula is split to the initial part $F_0$ and the learnt clauses $\mathit{Fl}$.
Since the input formula $F_0$ is fixed it is not a part of state anymore, but
rather an input parameter. The new component of the state --- the $\mathit{lnt}$
flag --- has a role in ensuring termination by preventing applying
\rulename{restart} and \rulename{forget} twice without learning a clause in
between. In addition, some changes in the rules ensure termination of some
variants of the system. Unit propagation is performed eagerly, i.e.,
\rulename{decide} is not applied when there is a unit clause present. Also,
backjumping is always performed to the minimal backjump level (Definition
\ref{def:minimalBackjumpLevel}). These stronger conditions are very often obeyed
in real \sat{} solver implementations, and so this system still makes their
faithful model.

\begin{defi}[State]
A state of the system is a five-tuple $(M,\mathit{Fl}, C, \mathit{cflct},
\mathit{lnt})$, where $M$ is a trail, $Fl$ is a formula, $C$ is a clause, and
$\mathit{cflct}$ and $\mathit{lnt}$ are Boolean variables. A state $([\,],
F_0, [\,], \bot, \bot)$ is a {\em initial state} for the input formula $F_0$.
\end{defi}

\begin{defi}[Transition rules]\hfill

\begin{small}
\noindent ${\fn decide}\ (M_1, \mathit{Fl}_1, C_1, \mathit{cflct}_1,
\mathit{lnt}_1)\ (M_2, \mathit{Fl}_2, C_2, \mathit{cflct}_2, \mathit{lnt}_2) \iff$
\vspace{-1mm}
\begin{eqnarray*}
\exists l.& &{\fn var}\ l\in \mathit{DecVars}\ \wedge\ l\notin M_1\ \wedge\ \opp{l}\notin M_1\ \wedge\ \\
& & \neg (\exists\ c\ l.\ c \in F_0\,@\,\mathit{Fl}_1 \ \wedge\ {\fn isUnitClause\ c\ l\ M_1})\ \wedge\\
& &M_2 = M_1\,@\,\decision{l} \ \wedge\ \mathit{Fl}_2 = \mathit{Fl}_1\ \wedge\ C_2 = C_1\ \wedge\ \mathit{cflct}_2 = \mathit{cflct}_1\ \wedge\ \mathit{lnt}_2 = \mathit{lnt}_1
\end{eqnarray*}

\noindent ${\fn unitPropagate}\ \ (M_1, \mathit{Fl}_1, C_1, \mathit{cflct}_1,
\mathit{lnt}_1)\ (M_2, \mathit{Fl}_2, C_2, \mathit{cflct}_2, \mathit{lnt}_2)\iff$
\vspace{-1mm}
\begin{eqnarray*}
\exists c\ l.& & c \in F_0\,@\,\mathit{Fl}_1\ \wedge\ {\fn isUnit}\ c\ l\ M_1\ \wedge\ \\
& &M_2 = M_1\,@\,\nondecision{l}\ \wedge\ \mathit{Fl}_2 = \mathit{Fl}_1\ \wedge\ C_2 = C_1 \ \wedge\ \mathit{cflct}_2 = \mathit{cflct}_1 \ \wedge\ \mathit{lnt}_2 = \mathit{lnt}_1
\end{eqnarray*}

\noindent ${\fn conflict}\ \ (M_1, \mathit{Fl}_1, C_1, \mathit{cflct}_1,
\mathit{lnt}_1)\ (M_2, \mathit{Fl}_2, C_2, \mathit{cflct}_2, \mathit{lnt}_2)\iff$
\vspace{-1mm}
\begin{eqnarray*}
\exists c.& &\mathit{cflct}_1 = \bot\ \wedge\ c \in F_0\,@\,\mathit{Fl}_1\ \wedge\ M_1\falsifies c \wedge\ \\
& &M_2 = M_1 \ \wedge\ \mathit{Fl}_2 = \mathit{Fl}_1\ \wedge\ C_2 = c \ \wedge\ \mathit{cflct}_2 = \top\ \wedge\ \mathit{lnt}_2 = \mathit{lnt}_1
\end{eqnarray*}

\noindent ${\fn explain}\ \ (M_1, \mathit{Fl}_1, C_1, \mathit{cflct}_1,
\mathit{lnt}_1)\ (M_2, \mathit{Fl}_2, C_2, \mathit{cflct}_2, \mathit{lnt}_2)\iff$
\vspace{-1mm}
\begin{eqnarray*}
\exists\ l\ c.& &\mathit{cflct}_1 = \top\ \wedge\ l\in C_1\ \wedge\ {\fn isReason}\
mc\ \opp{l}\ M_1\ \wedge\ c \in F_0\,@\,\mathit{Fl}_1\ \wedge\\
& &M_2 = M_1 \ \wedge\ \mathit{Fl}_2 = \mathit{Fl}_1\ \wedge\ C_2 = {\fn resolve}\ C_1\ c\ l \ \wedge\ \mathit{cflct}_2 = \top \ \wedge\ \mathit{lnt}_2 = \mathit{lnt}_1
\end{eqnarray*}

\noindent ${\fn backjumpLearn}\ (M_1, \mathit{Fl}_1, C_1, \mathit{cflct}_1,
\mathit{lnt}_1)\ (M_2, \mathit{Fl}_2, C_2, \mathit{cflct}_2, \mathit{lnt}_2)\iff$
\vspace{-1mm}
\begin{eqnarray*}
\exists c\ l\ level.& &\mathit{cflct}_1 = \top\ \wedge\ {\fn isMinimalBackjumpLevel}\ level\ l\ C_1\ M_1\ \wedge\ \\
& & M_2 = ({\fn prefixToLevel}\ level\ M_1)\,@\,\nondecision{l} \ \wedge\ \mathit{Fl}_2 = \mathit{Fl}_1 @ [C_1]\ \wedge\\
& & C_2 = [\,] \ \wedge\ \mathit{cflct}_2 = \bot \ \wedge\ \mathit{lnt}_2 = \top
\end{eqnarray*}

\noindent ${\fn forget}\ (M_1, \mathit{Fl}_1, C_1, \mathit{cflct}_1,
\mathit{lnt}_1)\ (M_2, \mathit{Fl}_2, C_2, \mathit{cflct}_2, \mathit{lnt}_2)\iff$
\vspace{-1mm}
\begin{eqnarray*}
\exists\ \mathit{Fc}.\ & & \mathit{cflct}_1 = \bot\ \wedge\ \mathit{lnt_1} = \top \\
 & &  \mathit{Fc} \subseteq \mathit{Fl} \ \wedge\ (\forall\ c \in \mathit{Fc}.\ \neg (\exists\ l.\ {\fn isReason}\ c\ l\ M_1)) \wedge\\
& & \mathit{Fl}_2 = \mathit{Fl}_1 \setminus \mathit{Fc}\ \wedge\ M_2 = M_1 \ \wedge\ C_2 = C_1 \ \wedge\ \mathit{cflct}_2 = \mathit{cflct}_1 \ \wedge\ \mathit{lnt}_2 = \bot
\end{eqnarray*}

\noindent ${\fn restart}\ (M_1, \mathit{Fl}_1, C_1, \mathit{cflct}_1,
\mathit{lnt}_1)\ (M_2, \mathit{Fl}_2, C_2, \mathit{cflct}_2, \mathit{lnt}_2) \iff$
\vspace{-1mm}
\begin{eqnarray*}
& & \mathit{cflct_1} = \bot \ \wedge\ \mathit{lnt_1} = \top \ \wedge\ \\
& & M_2 = {\fn prefixToLevel}\ 0\ M_1 \ \wedge\ \mathit{Fl}_2 = \mathit{Fl}_1\ \wedge\ C_2 = C_1\ \wedge\ \mathit{cflct}_2 = \mathit{cflct}_1 \ \wedge\ \mathit{lnt}_2 = \bot
\end{eqnarray*}
\end{small}
\end{defi}

These rules will be used to formulate three different transition systems. The
system $\rightarrowr$ consists of all rules except \rulename{restart}, the
system $\rightarrowf$ consists of all rules except \rulename{forget}, and the
system $\rightarrow$ consists of all rules.

\subsection{Properties}
\label{subsec:RestartProperties}

The structure of the invariants and the proofs of the properties of the system
are basically similar to those given in Section \ref{sec:conflictAnalysis},
while the termination proof requires a number of new insights.

\subsubsection{Invariants}
All invariants formulated so far hold, but the formula $F$, not present in the
new state, has to be replaced by $F_0\,@\,\mathit{Fl}$.

\subsubsection{Termination}

Termination of the system without restarts is proved first.

\begin{thm}[Termination for $\rightarrowr$]
If the set $\mathit{DecVars}$ is finite, for any formula $F_0$, the relation
$\rightarrowr$ is well-founded on the set of states $s$ such that
$s_0\rightarrowr^* s$, where $s_0$ is the initial state for $F_0$.
\end{thm}

\begin{proof}
Let $\succ$ be a (parametrized) lexicographic product (Definition
\ref{def:lexProd}), i.e., let

$
\succ\ \ \equiv\ \ \lexprod{\lexprodp{\lexprod{\succTrrestrict[\mathit{Vars}]}{\succBool}}{\left(\lambda
    s.\succC{M_s}\right)}}{\succBool},$ 

\noindent    
where $M_s$ is the trail in the state
$s$. By Proposition \ref{prop:wfLeksikografskaKombinacijaUredjenja} and Lemmas
\ref{lemma:wfTrailSuccRestricted}, \ref{lemma:wfConflictSucc}, and
\ref{lemma:wfClauseSucc}, the relation $\succ$ is well-founded.
If the state $(M_1,\mathit{Fl}_1, C_1, \mathit{cflct}_1, \mathit{lnt}_1)$
satisfies the invariants and if 
$(M_1,\mathit{Fl}_1, C_1, \mathit{cflct}_1,\mathit{lnt}_1) \rightarrowr 
(M_2,\mathit{Fl}_2, C_2, \mathit{cflct}_2,\mathit{lnt}_2),$ 
then 
$(M_1, \mathit{cflct}_1, C_1, \neg\mathit{lnt}_1) \succ (M_2, \mathit{cflct}_2, C_2, \neg\mathit{lnt}_2).$ 
Indeed, by Lemma
\ref{lemma:rightarrowsucc} the rules \rulename{unitPropagate}, \rulename{decide}
and \rulename{backjumpLearn} decrease $M$, the rule \rulename{conflict} does
not change $M$ but (by Lemma \ref{lemma:conflictSucc}) decreases
$\mathit{cflct}$, the rule \rulename{explain} does not change $M$ nor
$\mathit{cflct}$, but (by Lemma \ref{lemma:explainSucc}) decreases $C$, and the
rule \rulename{forget} does not change $M$, $\mathit{cflct}$, nor $C$, but
decreases $\neg \mathit{lnt}$.

From the above, the theorem holds by Proposition \ref{prop:wfInvImage} (for a 
suitable ${\fn f}$).
\end{proof}

The termination proof of the system without forgets is more involved.  We define
a (not necessarily well-founded) ordering of the formulae by inclusion and its
restriction with respect to the set of variables occurring in the formula.

\begin{defi}[$\succFIncl$]
$F_1\succFIncl F_2 \iff F_1 \subset F_2$.
\end{defi}

\begin{defi}[$\succFInclrestrict{\mathit{Vbl}}$]
$F_1 \succFInclrestrict{\mathit{Vbl}}\, F_2 \iff {\fn vars}\ F_1 \subseteq
  \mathit{Vlb}\ \wedge\ {\fn vars}\ F_1 \subseteq \mathit{Vbl}
  \ \wedge\ \overline{F_1} \succFIncl \overline{F_2},$ where $\overline{F}$
  denotes the formula obtained by removing duplicate literals from clauses and
  removing duplicate clauses.
\end{defi}

\begin{lem}
\label{lemma:wfFInclrestrict}
If the set $\mathit{Vbl}$ is finite, then the relation
$\succFInclrestrict{\mathit{Vbl}}$ is well-founded.
\end{lem}

The following lemma states that if unit propagation is done eagerly and if
backjumping is always performed to the minimal backjump level, then the clauses
that are learnt are always fresh, i.e., they do not belong to the current
formula.

\begin{lem}
\label{lemma:backjumpClauseIsNotPresent}
If $s_0$ is an initial state, $s_0 \rightarrowf^* s_A$ and ${\fn
  backjumpLearn}\ s_A\ s_B$, where $s_A = (M_A, \mathit{Fl}_A, C_A, \top,
\mathit{lnt}_A)$, then $C_A \notin F_0\,@\,\mathit{Fl}_A$.
\end{lem}

Therefore, \rulename{backjumpLearn} increases formula in the inclusion ordering.

\begin{lem}
\label{lemma:backjumpLearnSucc}
If $s_0 \rightarrowf s_A$ and ${\fn backjumpLearn}\ s_A\ s_B$ for initial
state $s_0$ and states $s_A$ and $s_B$, then $F_0\,@\,\mathit{Fl}_A
\succFInclrestrict{Vars} F_0\,@\,\mathit{Fl}_B$, where $F_A$ and $F_B$ are
formulae in states $s_A$ and $s_B$.
\end{lem}

\begin{thm}[Termination for $\rightarrowf$]
If the set $\mathit{DecVars}$ is finite, for any formula $F_0$, the relation
$\rightarrowf$ is well-founded on the set of states $s$ such that
$s_0\rightarrowf^* s$, where $s_0$ is the initial state for $F_0$.
\end{thm}

\begin{proof}
Let $\succ$ be a (parametrized) lexicographic product (Definition
\ref{def:lexProd}), i.e., let

\noindent
$\succ\ \ \equiv\ \ \lexprodp{\lexprod{\lexprod{\lexprod{\succFInclrestrict{\mathit{Vars}}}{\succBool}}{\succTrrestrict[\mathit{Vars}]}}{\succBool}}{\left(\lambda
  s. \succC{M_s}\right)},$

\noindent  
where $M_s$ is the trail in the state $s$.
By Proposition \ref{prop:wfLeksikografskaKombinacijaUredjenja} and Lemmas
\ref{lemma:wfTrailSuccRestricted}, \ref{lemma:wfConflictSucc},
\ref{lemma:wfClauseSucc}, and \ref{lemma:wfFInclrestrict}, the relation $\succ$
is well-founded. If the state $(M_1,\mathit{Fl}_1, C_1, \mathit{cflct}_1, \mathit{lnt}_1)$
satisfies the invariants and if 
$(M_1,\mathit{Fl}_1, C_1, \mathit{cflct}_1, \mathit{lnt}_1) \rightarrowf 
(M_1,\mathit{Fl}_1, C_1, \mathit{cflct}_1, \mathit{lnt}_1),$
then $(F_1, \neg\mathit{lnt}_1, M_1, \mathit{cflct}_1, C_1) \succ 
(F_2, \neg\mathit{lnt}_2, M_2, \mathit{cflct}_2, C_2).$ 
Indeed, by Lemma
\ref{lemma:backjumpLearnSucc} the rule \rulename{backjumpLearn} decreases
$F$, the rule \rulename{restart} does not change $F$ but decreases $\neg
\mathit{lnt}$, the rules \rulename{unitPropagate} and \rulename{decide} do not
change $F$ and $\mathit{lnt}$ but (by Lemma \ref{lemma:rightarrowsucc}) decrease
$M$, the rule \rulename{conflict} does not change $F$, $\mathit{lnt}$, nor $M$
but (by Lemma \ref{lemma:conflictSucc}) decreases $\mathit{cflct}$, and the rule
\rulename{explain} does not change $F$, $\mathit{lnt}$, $M$ nor
$\mathit{cflct}$, but (by Lemma \ref{lemma:explainSucc}) decreases $C$.

From the above, the theorem holds by Proposition \ref{prop:wfInvImage} (for a 
suitable ${\fn f}$).
\end{proof}

If both \rulename{forget} and \rulename{restart} are allowed, then the system
is not terminating.

\begin{thm}
The relation $\rightarrow$ is not well-founded on the set of states reachable
from the initial state.
\end{thm}

\proof
Consider the formula $[[-1, -2, 3],$ $[-1, -2, 4],$ $[-1, -3, -4],$ $[-5, -6,
 7],$ $[-5, -6, 8],$ $[-5, -7, -8]]$.  The following derivation chain (for
simplicity, not all components of the states are shown) proves that the relation
$\rightarrow$ is cyclic.

\medskip
\noindent
\begin{tabular}{l||l|l|l}
 rule                                                     & $M$                                                     & $\mathit{Fl}$             & $\mathit{lnt}$    \\ \hline
                                                          & $[\,]$                                                  & $[\,]$                    & $\bot$            \\
 \rulename{decide}, \rulename{decide}                     & $[\decision{1}, \decision{2}]$                                      & $[\,]$                    & $\bot$            \\
 \rulename{unitPropagate}, \rulename{unitPropagate}       & $[\decision{1}, \decision{2}, \nondecision{3}, \nondecision{4}]$                      & $[\,]$                    & $\bot$            \\
 \rulename{conflict}, \rulename{explain}, \rulename{explain}, \rulename{backjumpLearn}& $[\decision{1}, \nondecision{-2}]$                                   & $[[-1, -2]]$              & $\top$            \\
 \rulename{restart}                                       & $[\,]$                                                  & $[[-1, -2]]$              & $\bot$            \\
 \rulename{decide}, \rulename{decide}                     & $[\decision{5}, \decision{6}]$                                      & $[[-1, -2]]$              & $\bot$            \\
 \rulename{unitPropagate}, \rulename{unitPropagate}       & $[\decision{5}, \decision{6}, \decision{7}, \decision{8}]$                      & $[[-1, -2]]$              & $\bot$            \\
 \rulename{conflict}, \rulename{explain}, \rulename{explain}, \rulename{backjumpLearn}& $[\decision{5}, \decision{-6}]$                                   & $[[-1, -2], [-5, -6]]$    & $\top$            \\
 \rulename{forget}                                        & $[\decision{5}, \nondecision{-6}]$                                     & $[\,]$                    & $\bot$            \\
 \rulename{decide}, \rulename{decide}                     & $[\decision{5}, \nondecision{-6}, \decision{1}, \decision{2}]$                     & $[\,]$                    & $\bot$            \\
 \rulename{unitPropagate}, \rulename{unitPropagate}       & $[\decision{5}, \nondecision{-6}, \decision{1}, \decision{2}, \nondecision{3}, \nondecision{4}]$     & $[\,]$                    & $\bot$            \\
 \rulename{conflict}, \rulename{explain}, \rulename{explain}, \rulename{backjumpLearn}& $[\decision{5}, \nondecision{-6}, \decision{1}, \nondecision{-2}]$                  & $[[-1, -2]]$              & $\top$            \\
 \rulename{restart}                                       & $[\,]$                                                  & $[[-1, -2]]$              & $\bot$            \\
\end{tabular}
\medskip

Therefore, it holds that

\[([\,], [\,], [\,], \bot, \bot) \ \rightarrow^*\ ([\,], [[-1, -2]], [\,], \bot, \bot) \ \rightarrow^+\ ([\,], [[-1, -2]], [\,], \bot, \bot).\eqno{\qEd}\]\medskip

However, if there are additional restrictions on the rule application policy,
the system may be terminating. Since the number of different states for the
input formula $F_0$ is finite (when duplicate clauses and literals are removed),
there is a number $n_f$ (dependent on $F_0$) such that there is no chain of rule
applications without \rulename{forget} longer than $n_f$ ($\rightarrowf$ is
well-founded and therefore acyclic, so, on a finite set, there must exist $n_f$
such that $\rightarrowf^{n_f}$ is empty). Similarly, there is a number $n_r$
(dependent on $F_0$) such that there is no chain of rule applications without
\rulename{restart} longer than $n_r$.  So, termination is ensured for any policy
that guarantees that there is a point where the application of \rulename{forget}
will be forbidden for at least $n_f$ steps or that there is a point where the
application of \rulename{restart} will be forbidden for at least $n_r$ steps.

\subsubsection{Soundness, Completeness and Correctness}

Soundness and completeness proofs from previous sections hold with minor
modifications necessary to adapt them to the new definition of state and
rules. The most demanding part is to update Lemma \ref{lemma:invariantsHold}
and to prove that the new rules maintain the invariants.

\section{Related Work and Discussions}
\label{sec:discuss}

The original \dpll{} procedure \cite{dll62} has been described in many logic
textbooks, along with informal proofs of its correctness (e.g.,
\cite{davis-tcs}).  First steps towards verification of modern \dpll-based
\sat{} solvers have been made only recently. Zhang and Malik have informally
proved correctness of a modern \sat{} solver \cite{zhangmalik-validating}. Their
proof is very informal, the specification of the solver is given in pseudo-code
and it describes only one strategy for applying rules. The authors of two
abstract transition systems for \sat{} also give correctness proofs
\cite{NieOT-Jacm06, Krstic-Frocos07}. These specifications and the proofs are
much more formal than those given in \cite{zhangmalik-validating}, but they are
also not machine-verifiable and are much less rigorous than the proofs presented
in this paper.

In recent years, several machine-verifiable correctness proofs for \sat{}
solvers were constructed. Lescuyer and Conchon formalized, within Coq, a \sat{}
solver based on the classical \dpll{} procedure and its correctness proof
\cite{sat-verification-francuzi}. They used a deep embedding, so this approach
enables execution of the \sat{} solver in Coq and, further, a reflexive tactic.
Mari\' c and Jani\v ci\' c formalized a correctness proof for the classical
\dpll{} procedure by shallow embedding into Isabelle/HOL
\cite{Informatica}. Shankar and Vaucher formally and mechanically verified a
high-level description of a modern \dpll{}-based \sat{} solver within the system
PVS \cite{shankar-vaucher}. However, unlike this paper which formalizes abstract
descriptions for \sat{}, they formalize a very specific \sat{} solver
implementation within PVS. Mari\' c proved partial correctness (termination was
not discussed) of an imperative pseudo-code of a modern \sat{} solver using
Hoare logic approach \cite{JAR} and total correctness of a \sat{} solver
implemented in Isabelle/HOL using shallow embedding \cite{TCS}.  Both these
formalizations use features of the transition systems described in this paper
and provide links between the transition systems and executable implementations
of modern \sat{} solvers. In the former approach, the verified specification can
be rewritten to an executable code in an imperative programming
language\footnote{As done in the implementation of our \sat{} solver ArgoSAT.}
while in the latter approach, an executable code in a functional language can be
exported from the specification by automatic means \cite{code-generation}.

The transition system discussed in Section \ref{sec:DPLLSearch} corresponds to a
non-recursive version of the classical \dpll{} procedure. The transition systems
and correctness proofs presented in the later sections are closely related to
the systems of Nieuwenhuis et al.~\cite{NieOT-Jacm06} and Krsti\' c and Goel
\cite{Krstic-Frocos07}. However, there are some significant differences, both in
the level of precision in the proofs and in the definitions of the rules.

Informal (non machine-verifiable) proofs allow authors some degree of
imprecision.  For example, in \cite{NieOT-Jacm06} and \cite{Krstic-Frocos07}
clauses are defined as ``disjunctions of literals'' and formulae as
``conjunctions of clauses'', and this leaves unclear some issues such as whether
duplicates are allowed. The ordering of clauses and literals is considered to be
irrelevant --- in \cite{Krstic-Frocos07} it is said that ``clauses containing
the same literals in different order are considered equal'', and in
\cite{NieOT-Jacm06} it is not explicitly said, but only implied (e.g., clauses
in the \rulename{unitPropagate} rule are written as $C \vee l$, where
$M\falsifies C$ and $l$ is undefined in $M$, and from this it is clear that the
order of literals must be irrelevant, or otherwise only last literals in clauses
could be propagated). Therefore, clauses and formulae are basically defined as
sets or multisets of literals. In our formal definition, clauses and formulae
are defined as lists. Although a choice whether to use lists, multisets, or sets
in these basic definitions might not seem so important, fully formal proofs show
that this choice makes a very big difference. Namely, using sets saves much
effort in the proof. For example, if formulae may contain repeated clauses, easy
termination arguments like ``there are finitely many different clauses that can
be learnt'' cannot be applied. On the other hand, using sets makes the systems
quite different from real \sat{} solver implementations --- eliminating
duplicates from clauses during solving is possible and cheap, but explicitly
maintaining absence of duplicate clauses from formulae may be intolerably
expensive. It can be proved that maintaining absence of duplicate clauses can
be, under some conditions on the rules, implicitly guaranteed only by
eliminating duplicate clauses from formulae during initialization. Solvers
typically assume this complex fact, but it was not proved before for formulae
represented by lists, while for systems using sets this issue is irrelevant.

The system given in \cite{NieOT-Jacm06} is very close to the system given in
Section \ref{sec:backjumping} and later extended in Section
\ref{sec:learningForgetting}. The requirement that the set of decision literals
exactly coincides with the set of literals from the input formula is too strong
and is not always present in real \sat{} solvers, so it is relaxed in our system
and the set $\mathit{DecVars}$ is introduced (a similar technique is used in
\cite{Krstic-Frocos07}). Also, the definition of the \rulename{backjump} rule
from \cite{NieOT-Jacm06} requires that there is a false clause in the formula
being solved when the rule is applied, but our formal analysis of the proofs
shows that this assumption is not required, so it is omitted from Definition
\ref{def:backjumping_system}. As already mentioned, the condition that the unit
clauses belong to the formula is also relaxed, and propagating can be performed
over arbitrary consequences of the formula. The invariants used in the proofs
and the soundness proof are basically the same in \cite{NieOT-Jacm06} and in
this paper, but the amount of details had to be significantly increased to reach
a machine-verifiable proof. Our completeness proof is somewhat simpler. The
ordering used in termination proof for the system with backjumping in
\cite{NieOT-Jacm06} expresses a similar idea to ours, but is much more
complex. A conflict analysis process is not described within the system from
\cite{NieOT-Jacm06}.

The system given in \cite{Krstic-Frocos07} is close to the system given in
Section \ref{sec:conflictAnalysis}, with some minor differences. Namely, in our
system, instead of a set of decision literals, the set of decision variables is
considered. Also, unit, conflict and reason clauses need not be present in the
formula. The conflict set used in \cite{Krstic-Frocos07} along with its
distinguished value ${\fn no\_cflct}$ is here replaced by the conflict flag and
a conflict clause (the conflict set is the set of opposites of literals
occurring in our conflict clauses). The underlying reasoning used in two total
correctness proofs is the same, although in \cite{Krstic-Frocos07} the
invariants are not explicitly formulated and the proof is monolithic (lemmas are
not present) and rather informal.

Formalization of termination proofs from both \cite{NieOT-Jacm06} and
\cite{Krstic-Frocos07} required the greatest effort in the
formalization. Although arguments like ``between any two applications of the
rule \ldots there must be an occurrence of the rule \ldots'', heavily used in
informal termination proofs, could be formalized, we felt that constructing
explicit termination orderings is much cleaner.

In \cite{Krstic-Frocos07} termination of systems with restarts is not thoroughly
discussed and in \cite{NieOT-Jacm06} it is proved very informally, under a
strong condition that the periodicity of restarts is strictly increasing. This
is often not the case in many modern \sat{} solver implementations. In this
paper, we have (formally) proved that restarting can be performed very
frequently (after each conflict) without compromising total correctness.
However, some additional requirements (unit propagation must be exhaustive,
backjumping must be performed to minimal backjumping levels, and backjump lemmas
must always be learnt) are used in the proof, but these are always present in
modern \sat{} solvers. Although the issue has been addressed in the literature,
we are not aware of a previous proof of termination of frequent restarting.

\section{Conclusions}
\label{sec:conclusions}

We presented a formalization of modern \sat{} solvers and their properties in
the form of \emph{abstract state transition systems}. Several different \sat{}
solvers are formalized --- from the classical \dpll{} procedure to its modern
successors. The systems are defined in a very abstract way so they cover a wide
range of \sat{} solving procedures. The formalization is made within the
Isabelle/HOL system and the total correctness properties (soundness,
termination, completeness) are shown for each presented system.

Central theorems claim (roughly) that a transition system, i.e., a \sat{}
solver, terminates and returns an answer {\em yes} if and only if the input
formula is satisfiable. This whole construction boils down to the simple
definition of satisfiable formula, which can be confirmed by manual inspection.

Our formalization builds up on the previous work on state transition systems for
\sat{} and also on correctness arguments for other \sat{} systems. However, our
formalization is the first that gives machine-verifiable total correctness
proofs for systems that are close to modern \sat{} solvers. Also, compared to
other abstract descriptions, our systems are more general (so can cover a wider
range of possible solvers) and require weaker assumptions that ensure the
correctness properties. Thanks to the framework of formalized mathematics, we
explicitly separated notions of soundness and completeness, and defined all
notions and properties relevant for \sat{} solving, often neglected to some
extent in informal presentations.

Our experience in the \sat{} verification project shows that having imperative
software modelled abstractly, in the form of abstract state transition systems,
makes the verification cleaner and more flexible. It can be used as a key
building block in proving correctness of \sat{} solvers by using other
verification approaches which significantly simplifies the overall verification
effort.

\section*{Acknowledgement}
This work was partially supported by the Serbian Ministry of Science grant
174021 and by the SNF grant SCOPES IZ73Z0\_127979/1. We are grateful to
Natarajan Shankar for sharing with us his unpublished manuscript
\cite{shankar-vaucher}.  We are also grateful to anonymous reviewers for very
careful reading and for detailed and useful comments on an earlier version of
this paper.

\begin{small}
\newcommand{\etalchar}[1]{$^{#1}$}

\end{small}

\end{document}